\documentclass{amsart}
\pdfoutput=1

\usepackage{lipsum}
\usepackage{amsfonts}
\usepackage{graphicx}
\usepackage{epstopdf}
\usepackage{algorithmic}
\ifpdf
  \DeclareGraphicsExtensions{.eps,.pdf,.png,.jpg}
\else
  \DeclareGraphicsExtensions{.eps}
\fi

\usepackage{amsopn}

\usepackage[square,sort,comma,numbers]{natbib}
\usepackage{graphicx,amsmath,amssymb}

\usepackage{ccg}

\usepackage{xargs}
\usepackage[colorinlistoftodos,prependcaption,textsize=small,backgroundcolor=lightgray]{todonotes}

\usepackage{tikz}
\usetikzlibrary{matrix,positioning,shapes,arrows,fit,backgrounds,automata,math}

\usepackage{enumitem}

\usepackage{caption}
\usepackage{subcaption}

\allowdisplaybreaks

\usepackage[utf8]{inputenc}
\usepackage[english]{babel}

\makeatletter
\useshorthands{"}%
\defineshorthand{"-}{\nobreak-\bbl@allowhyphens}
\defineshorthand{"=}{\nobreak--\bbl@allowhyphens}
\defineshorthand{"~}{\textormath{\leavevmode\hbox{-}}{-}}
\makeatother

\newtheorem{theorem}{Theorem}[section]
\newtheorem{lemma}[theorem]{Lemma}
\newtheorem{definition}[theorem]{Definition}
\newtheorem{corollary}[theorem]{Corollary}
\newtheorem{example}[theorem]{Example}

\numberwithin{equation}{section}

\newcommand*{\ROT}[1]{\raisebox{\depth}{\scalebox{1}[-1]{\hbox{$#1\mathstrut$}}}}
\newcommand*{\SHORTWORDS}{L_1}

\newcommand*{\FORESTS}{\mathcal{F}}
\newcommand*{\RULETREES}{\mathcal{R}}
\newcommand*{\RULESET}{\mathcal{R}}
\newcommand*{\DEVTREES}{\mathcal{D}}

\newcommand*{\TREELANG}{T}
\newcommand*{\STRINGLANG}{L}
\newcommand*{\LEXICON}{\mathcal{L}}
\newcommand*{\SPINES}{{\mathcal S}}

\newcommand*{\ARGUMENTS}[2][]{\mathcal{A}(#2\ifx\\#1\\\else, #1\fi)}
\newcommand*{\CATEGORIES}[2][]{\mathcal{C}(#2\ifx\\#1\\\else, #1\fi)}
\newcommand*{\FOCATEGORIES}[2][]{\mathcal{C}_{f}(#2\ifx\\#1\\\else, #1\fi)}
\newcommand*{\HCATEGORIES}[3][]{\mathcal{C}^{#2}(#3\ifx\\#1\\\else, #1\fi)}
\newcommand*{\CATEGORYSEQ}{\mathcal{C}_{\mathcal G}}
\newcommand*{\CCG}{CCG}
\newcommand*{\CONTEXTS}[1]{C_{#1}}
\newcommand*{\DERIVES}{\Rightarrow}

\newcommand*{\POS}{\mathrm{pos}}
\newcommand*{\LEAVES}{\mathrm{leaves}}

\makeatother

\newcommand*{\TREES}[1]{T_{#1}}

\newcommand*{\NONTERMINALS}{N}
\newcommand*{\TERMINALS}{\Sigma}

\DeclareMathOperator*{\TARGET}{tar}

\newcommand*{\AUTOMATON}{\mathcal{A}}

\DeclareMathOperator*{\SECONDARY}{sec}
\DeclareMathOperator*{\Pref}{Pref}

\newcommand*{\GENERATOR}{\mathrm{gen}}

\DeclareMathOperator*{\ARITY}{ar}
\DeclareMathOperator*{\ARGUMENT}{arg}

\newcommand*{\STATEOUTPUT}{\tau}
\newcommand*{\STATEOUTPUTEXTEND}{\tau'}

\newcommand*{\SLASHDIR}{\mathrm{slash}}
\newcommand*{\COMB}{\mathrm{comb}}

\newcommand*{\TOPCATL}{\top_{\!\!\mathcal A}}
\newcommand*{\TOPCATS}{\top_{\!\!L_1}}
\newcommand*{\RESTRCAT}{C_{\mathrm{wf}}}

\newcommand*{\ANNOA}{\frac{\alpha_2}{\overline{a} \enspace {s}}}
\newcommand*{\ANNOB}{\frac{\beta_2}{{s}  \enspace \overline{b}}}
\newcommand*{\ANNOC}{\frac{\gamma_2}{{s} \enspace \overline{c}}}
\newcommand*{\ANNOE}{\frac{\eta_2}{\overline{e}   \enspace \overline{b}}}

\providecommand*{\nat}[0]{\ensuremath{\mathbb{N}}}
\providecommand*{\natp}[0]{\ensuremath{\mathbb{N}_{\scriptscriptstyle+}}}
\providecommand*{\seq}[3]{\ensuremath{#1_{#2}, \dotsc, #1_{#3}}}
\providecommand*{\word}[3]{\ensuremath{#1_{#2} \dotsm #1_{#3}}}
\providecommand*{\abs}[1]{\ensuremath{\vert #1 \rvert}}
\providecommand*{\SBox}[0]{\ensuremath{\mathchoice%
    {{\scriptstyle \Box}}%
    {{\scriptstyle \Box}}%
    {{\scriptscriptstyle \Box}}%
    {{\scriptscriptstyle \Box}}%
}}
\providecommand*{\Pplus}[0]{\ensuremath{\mathcal P_{\!\!\scriptscriptstyle+}}}

\DeclareMathOperator{\pos}{pos}
\DeclareMathOperator{\type}{cat}

\DeclareMathOperator{\attach}{attach}
\DeclareMathOperator{\Next}{Next}
\DeclareMathOperator{\RETURN}{pop}
\DeclareMathOperator{\yield}{yield}

\makeatletter
\renewcommand*\env@matrix[1][c]{\hskip -\arraycolsep
  \let\@ifnextchar\new@ifnextchar
  \array{*\c@MaxMatrixCols #1}}
\makeatother

\def\redpos{2.5}
\def\redcol{black}

\def\redposside{(3,-1.25)}

\tikzset{terminalrule/.code n args={2}{
\node at (0,-0.5) {#1}[level distance=2em];
\draw[->] (0.5,-0.5) -- (1,-0.5);
\node at (1.5,-0.5) {#2};
}}

\tikzset{siderulel/.code n args={4}{
\node at (0,-0.5) {#1}[level distance=2em]
    child {node {$\SBox$}};
\draw[->] (0.5,-0.75) -- (1,-0.75);
\node at (1.75,-0.5) {#2} [grow'=down,level distance=2em,sibling distance=2em]
    child {node {$\SBox$}}
    child {node {#3}}
    ;
\node [\redcol] at \redposside {#4} [grow'=down,\redcol,level
distance=2em];
}}

\tikzset{sideruler/.code n args={4}{
\node at (0,-0.5) {#1}[level distance=2em]
    child {node {$\SBox$}};
\draw[->] (0.5,-0.75) -- (1,-0.75);
\node at (1.75,-0.5) {#2} [grow'=down,level distance=2em,sibling distance=2em]
    child {node {#3}}
    child {node {$\SBox$}}
    ;
\node [\redcol] at \redposside {#4} [grow'=down,\redcol,level
distance=2em];
}}

\tikzset{spinestart/.code n args={5}{
\node at (0,-0.5) {#1}[level distance=2em];
\draw[->] (0.5,-0.5) -- (1,-0.5);
\node at (1.5,0) {#2} [grow'=down,level distance=2em]
    child {node {#3}
    };
\node [\redcol] at (\redpos,0) {#4} [grow'=down,\redcol,level distance=2em]
    child {node {#5} edge from parent [draw=none]
    };
}}

\tikzset{spinestart2/.code n args={6}{
\node at (0,-0.5) {#1}[level distance=2em];
\draw[->] (0.5,-0.5) -- (1,-0.5);
\node at (1.5,0) {#2} [grow'=down,level distance=2em]
    child {node {#3}
    };
\node [\redcol] at (\redpos,0) {#4} [grow'=down,\redcol,level distance=2em]
    child {node {#5} edge from parent [draw=none]
    };
\node [\redcol] at (0.3,0) {#6};
}}

\tikzset{chain2/.code n args={5}{
\node at (0,-0.25) {#1}[level distance=2em]
    child {node {$\SBox$}};
\draw[->] (0.5,-0.5) -- (1,-0.5);
\node at (1.5,0) {#2} [grow'=down,level distance=2em]
    child {node {#3}
            child {node {$\SBox$}}
    };
\node [\redcol] at (\redpos,0) {#4} [grow'=down,\redcol,level distance=2em]
    child {node {#5} edge from parent [draw=none]
    };
}}

\tikzset{chain3/.code n args={7}{
\node at (0,-0.5) {#1}[level distance=2em]
    child {node {$\SBox$}};
\draw[->] (0.5,-0.75) -- (1,-0.75);
\node at (1.5,0) {#2} [grow'=down,level distance=2em]
    child {node {#3}
        child {node {#4}
            child {node {$\SBox$}}
        }
    };
\node [\redcol] at (\redpos,0) {#5} [grow'=down,\redcol,level distance=2em]
    child {node {#6} edge from parent [draw=none]
        child {node {#7} edge from parent [draw=none]}
    };
}}

\tikzset{chain4/.code n args={9}{
\node at (0,-1) {#1}[level distance=2em]
    child {node {$\SBox$}};
\draw[->] (0.5,-1.25) -- (1,-1.25);
\node at (1.5,0) {#2} [grow'=down,level distance=2em]
    child {node {#3}
        child {node {#4}
            child {node {#5}
                child {node {$\SBox$}}
            }
        }
    };
\node [\redcol] at (\redpos,0) {#6} [grow'=down,\redcol,level distance=2em]
    child {node {#7} edge from parent [draw=none]
        child {node {#8} edge from parent [draw=none]
            child {node {#9} edge from parent [draw=none]}
        }
    };
}}

\begin{document}

\title{Strong Equivalence of TAG and CCG}
\thanks{
This is a revised and extended version of \emph{Strong
Equivalence of TAG and CCG}, appearing in
\emph{Transactions of the Association for Computational Linguistics 9 (2021):
pp.\ 707-720}.}

\author[A.~Maletti]{Andreas Maletti}
\address{Faculty of Mathematics and Computer Science,
Universität Leipzig, P.O. box 100\,920, D-04009 Leipzig, Germany}
\email{\{maletti,schiffer\}@informatik.uni-leipzig.de}

\author[L.~K.~Schiffer]{Lena Katharina Schiffer}
\thanks{The work of Lena Katharina Schiffer was funded by the German Research
Foundation (DFG) Research Training Group GRK 1763 'Quantitative Logics and
Automata'.}



\begin{abstract}
  Tree-adjoining grammar (TAG) and combinatory categorial grammar
  (CCG) are two well-established mildly context-sensitive grammar
  formalisms that are known to have the same expressive power on
  strings (i.e., generate the same class of string languages).  It is
  demonstrated that their expressive power on trees also essentially
  coincides.  In fact, CCG without lexicon entries for the empty
  string and only first-order rules of degree at most 2 are sufficient
  for its full expressive power.
\end{abstract}

\maketitle

\section{Introduction}

Combinatory categorial grammar
(CCG)~\cite{steedman2000syntactic,steedman2011combinatory} is one of
several grammar formalisms that were introduced as an extension of
context-free grammars.  In particular, CCG extends the classical
categorial grammar~\cite{bargaisha60}, which has the same expressivity
as context-free grammar, by rules that are inspired by combinatory
logic~\cite{curry1958combinatory}.
CCG is a mildly context-sensitive grammar formalism~\cite{jos85},
which are formalisms that are efficiently parsable (i.e., in
polynomial time) and have expressivity beyond the context-free
languages.
Additionally, they are able to express a limited amount of cross-serial
dependencies and have the
constant growth property.
Due to these features and its notion of syntactic categories,
which is quite intuitive for natural languages, CCG has become widely
applied in compuational linguistics.
Further, it can be enhanced by semantics through the lambda
calculus~\cite{steedman2000syntactic}.

CCG is based on a lexicon and a rule system.  The lexicon assigns
syntactic categories to the symbols of an input string and the rule
system describes how neighboring categories can be combined to new 
categories.  Each category has a target, which is similar to the
return type of a function, and optionally, a number of
arguments.  Different from functions, each argument has a
directionality that indicates if it is expected on the left or the
right side.  If repeated combination of categories leads to
a (binary) derivation tree that comprises all input symbols and is
rooted in an initial category, then the input string is accepted.

When defining CCG, there are many degrees of freedom yielding a number
of different variants~\cite{steedman2000syntactic,bal02,steedman2011combinatory,kuhlmann2015lexicalization}.
This is a consequence of the linguistically motivated need to easily
express specific structures that have been identified in a particular
theory of syntax for a given natural language.  However, we and
others~\cite{kuhlmann2015lexicalization} are interested in the
expressive power of CCG as generators of formal languages, since this
allows us to disentangle the confusion of subtly different formalisms
and identify the principal structures expressible by a common core of
the formalisms.  As linguistic structure calls for a representation
that goes beyond strings, we aim for a characterization of expressive
power in terms of the generated trees.

The most famous result on the expressive power of CCG is
by~Vijay"-Shanker and Weir~\cite{vijayshanker1994equivalence} and
shows that tree-adjoining grammar~(TAG), linear indexed grammar~(LIG),
head grammar~(HG), and CCG generate the same string languages.  An 
equivalent automaton model is the embedded push-down
automaton~\cite{vijayashanker1988study}.  In the definition of CCG
used by Vijay"-Shanker and Weir~\cite{vijayshanker1994equivalence},
the lexicon allows $\varepsilon$-entries, which assign syntactic
categories to the empty string~$\varepsilon$.  Their rule system
restricts rules to specific categories and limits the rule degree.
CCG with unbounded rule degree are
Turing-complete~\cite{kuhsatjon18}.
Prefix-closed CCG without
target restrictions, in which the rules obey special closure
properties, are less powerful. This even holds for multimodal
CCG~\cite{kuhlmann2010importance,kuhlmann2015lexicalization}, which
allow many types of directionality indicators (i.e.~slashes beyond
those for left and right).

When going beyond the level of string languages, there exist different
notions of strong generative power.  We regard two formalisms as
strongly equivalent if their generated derivation tree languages
coincide.
Oftentimes, we will consider strong equivalence modulo relabeling.
For example, the well-known local and regular tree
grammars~\cite{gecste97} are strongly equivalent modulo relabeling.
On the other hand, Hockenmaier and
Young~\cite{hockenmaier2008nonlocal} regard two formalisms as strongly
equivalent if they capture the same sets of dependencies.  Then there
exist specific scrambling cases whose dependencies can be expressed by
their CCG, but not by Lexicalized TAG (LTAG).   Their CCG are
syntactically more expressive than ours and allow type-raising,
whereas the strong generative capacity (in our sense) of LTAG is
strictly smaller than that of TAG~\cite{kuhsat12}.  The dependencies
expressed by CCG without rule restrictions and TAG are shown to be
incomparable by Koller and Kuhlmann~\cite{koller2009dependency}.
It has been shown that CCG is able to generate exactly the
\emph{separable permutations} (i.e.~permutations that can label the
leaves of a binary tree such that the leaf labels of each subtree are
a set of consecutive elements of the original order) of a specific
``natural order of dominance'', while TAG can also express
non"-separable permutations~\cite{stanojevic2021formal}.

Returning to our notion of strong generative capacity,
Kuhlmann, Maletti, and Schiffer~\cite{kuhlmann2019treegenerative,kuhmalsch2021}
investigated the tree-generative capacity of CCG without
$\varepsilon$"~entries.  The generated trees are always binary.  CCG
with application and first-degree composition rules generate exactly
the regular tree languages~\cite{gecste97}.  Without the composition
rules, only a proper subset can be generated.  The languages of CCG
rule trees (i.e.\@ trees labeled by applied rules instead of
categories) with bounded rule degree can also be generated by simple
monadic context-free tree grammar~(sCFTG).

For the converse direction we show that the tree languages generated
by sCFTG can also be generated by CCG, which shows strong equivalence.
This answers several open questions.  Since sCFTG and TAG are strongly
equivalent \cite{kepser2011equivalence}, our result also shows strong  
equivalence of CCG~and~TAG.  In contrast to the construction of
Vijay"-Shanker and Weir~\cite{vijayshanker1994equivalence}, which
relies heavily on $\varepsilon$"~entries, our construction avoids them
and shows that they do not increase the expressive power of CCG.
Additionally, we only use rules up to degree~$2$ and first-order
categories (i.e., arguments are atomic), which shows that larger rule
degree or higher-order categories do not increase the expressive
power.

Our construction proceeds roughly as follows.  We begin with a spine
grammar, which is a variant of sCFTG that is also strongly equivalent
to TAG (up to relabeling).  We encode its spines using a context-free
grammar, which in turn can be represented by a special variant of
push-down automata.  Finally, the runs of the push-down automaton are
simulated by a CCG such that the stack operations of the automaton are
realized by adding and removing arguments of the categories.

\section{Preliminaries}
The nonnegative integers are~$\nat$ and the positive integers
are~$\natp$.  For every $k \in \nat$, we let $[k] = \{i \in \nat \mid
1 \leq i \leq k\}$.  Given a set $A$, let $\mathcal P(A) = \{ A' \mid
A' \subseteq A \}$ be the power"-set of $A$. As usual, $\pi_i \colon
X_1 \times \cdots \times X_n \to X_i$ projects a tuple to its $i$"-th
component and is given by $\pi_i(\langle x_1,\dots,x_n \rangle) =
x_i$, where each $X_j$ with $j \in [n]$ is a set.  An \emph{alphabet}
is a finite set of symbols.  The set~$\Sigma^*$ contains all strings
over the alphabet~$\Sigma$ including the empty string~$\varepsilon$.
We let $\Sigma^{\scriptscriptstyle +} = \Sigma^* \setminus
\{\varepsilon\}$.  The length of $w \in \Sigma^*$ is~$\abs w$, and
concatenation is written as juxtaposition.  The
\emph{prefixes}~$\Pref(w)$ of a string $w \in \Sigma^*$ are $\{ u \in
\Sigma^* \mid \exists v \in \Sigma^* \colon w = uv\}$.  A \emph{string
  language} is a subset $L \subseteq \Sigma^*$.  Given a
relation~$\mathord{\Rightarrow} \subseteq S^2$, we let $\Rightarrow^*$
be the reflexive, transitive closure of~$\Rightarrow$.

\subsection{Representations of String Languages}

We briefly recall three standard formalisms for representing string
languages. We start with nondeterministic finite
automata~\cite{hopull79}.

\begin{definition}
  A \emph{nondeterministic finite automaton} (NFA) $\AUTOMATON = (Q,
  \Sigma, \delta, I, F)$ is a tuple consisting of
  \begin{enumerate}[label=(\roman*)]
  \item finite sets~$Q$ and $\Sigma$  of \emph{states} and
  \emph{input symbols}, respectively,
  \item a \emph{transition relation} $\delta \subseteq Q \times \Sigma
    \times Q$, and
  \item sets~$I, F \subseteq Q$ of \emph{initial} and \emph{final states},
  respectively.
  \end{enumerate}
\end{definition}
The transition relation can be extended to a function~$\hat{\delta}
\colon Q \times \Sigma^{*} \to \mathcal P(Q)$ given by
\begin{align*}
  \hat{\delta}(q, \varepsilon) = \{ q \} \qquad  \qquad
  \hat{\delta}(q, aw) = \bigcup_{(q,a,p) \in \delta} \hat{\delta}(p,w)
\end{align*}
for all $q \in Q$, $a \in \Sigma$, and $w \in \Sigma^{*}$.  The
language \emph{accepted} by a given NFA~$\AUTOMATON$ is defined as 
\begin{align*}
  L(\AUTOMATON) = \bigcup_{q \in I} \bigl\{ w \in
  \Sigma^{*} \mid \hat\delta (q, w) \cap F \neq \emptyset \bigr\} \enspace.
\end{align*}
Given a string language~$L$, if there exists a NFA~$\AUTOMATON$ with
$L(\AUTOMATON) = L$, then we call $L$ \emph{regular}.

Next, let us recall context-free grammars~\cite{ber79}.
\begin{definition}
A \emph{context-free grammar} (CFG) $\mathcal G = (N, \Sigma, S, P)$ 
consists of 
\begin{enumerate}[label=(\roman*)]
\item disjoint finite sets $N$ and $\Sigma$ of \emph{nonterminal} and
\emph{terminal symbols}, respectively,
\item a \emph{start nonterminal} $S \in N$, and
\item a finite set $P \subseteq N \times (N \cup \Sigma)^*$ of
  \emph{productions}.
\end{enumerate}
\end{definition}
In the following let $\mathcal G = (N, \Sigma, S, P)$ be a CFG. We
write productions $(n,r)$ as $n \to r$.  Given $n \to r \in P$ and
$u,v \in (N \cup \Sigma)^*$, we write $unv \DERIVES_{\mathcal G} urv$
and say that $unv$ \emph{derives} $urv$.  The language
\emph{generated} by $\mathcal G$ is $L(\mathcal G) =  \{ w \in
\Sigma^* \mid S \DERIVES_{\mathcal G}^* w \bigr\}$.  Given a string
language $L$, if there exists a CFG~$\mathcal G$ with $L(\mathcal G) =
L$, then we call $L$ \emph{context-free}.

Finally, let us recall push-down automata~\cite{autberboa97},
which accept the ($\varepsilon$-free) context-free languages.
For any alphabet~$\Sigma$ and special symbol~$\bot
\notin \Sigma$, we let $\Sigma_\bot = \Sigma \cup \{\bot\}$.
Moreover, we let~$\Sigma^{\leq 1}$ be the set of strings of length at
most~$1$ over the alphabet~$\Sigma$; i.e.,~$\Sigma^{\leq 1} =
\{\varepsilon\} \cup \Sigma$.

\begin{definition}
\label{def:pda}
  A \emph{push-down automaton}~(PDA) $\AUTOMATON = (Q, \Sigma, \Gamma,
  \delta, I, \bot, F)$ is a tuple consisting of
  \begin{enumerate}[label=(\roman*)]
  \item finite sets~$Q$, $\Sigma$, and $\Gamma$ of \emph{states},
  \emph{input symbols}, and \emph{stack symbols}, respectively,
  \item a set~$\delta \subseteq \bigl(Q \times \Sigma \times
    \Gamma_\bot^{\leq 1} \times \Gamma^{\leq 1} \times Q \bigr)
    \setminus \bigl(Q \times \Sigma \times \Gamma_\bot \times \Gamma
    \times Q \bigr)$ of \emph{transitions},
  \item sets~$I, F \subseteq Q$ of \emph{initial} and \emph{final
      states}, respectively, and
  \item an initial stack symbol~$\bot \notin \Gamma$.
  \end{enumerate}
\end{definition}
Given a PDA~$\AUTOMATON = (Q, \Sigma, \Gamma, \delta, I, \bot, F)$,
let~$\textrm{Conf}_\AUTOMATON = Q \times \Sigma^* \times \Gamma_\bot^*$
be the set of \emph{configurations}.
Intuitively speaking, in configuration~$\langle q, w, \alpha \rangle
\in \textrm{Conf}_\AUTOMATON$ the PDA~$\AUTOMATON$ is in state~$q$
with stack contents~$\alpha$ and still has to read the input
string~$w$.  The \emph{move relation}~$\mathord{\vdash_{\mathcal A}}
\subseteq \textrm{Conf}_\AUTOMATON^2$ is defined as follows:
\[
  \mathord{\vdash_{\AUTOMATON}} =
  \bigcup_{\substack{(q, a, \gamma, \gamma', q') \in \delta,\\
  w \in \Sigma^*,\, \alpha \in \Gamma^*_\bot}}
  \Bigl\{ \bigl(\langle q, aw, \gamma
  \alpha\rangle,\, \langle q', w, \gamma' \alpha\rangle \bigr)
  \in \textrm{Conf}_\AUTOMATON^2 \mid \gamma\alpha \neq \varepsilon
  \Bigr\} \enspace. \]
The configuration~$\langle q, w, \alpha\rangle$ is \emph{initial}
(respectively, \emph{final}) if
$q \in I$, $w \in \Sigma^{\scriptscriptstyle +}$, and $\alpha = \bot$
(respectively, $q \in F$, $w = \varepsilon$, and $\alpha =
\varepsilon$).
An \emph{accepting run} is a sequence~$\seq \xi0n \in
\textrm{Conf}_\AUTOMATON$ of configurations that are successively
related by moves (i.e., $\xi_{i-1} \vdash_\AUTOMATON \xi_i$ for all~$i
\in [n]$), that starts with an initial configuration~$\xi_0$, and
finishes in a final configuration~$\xi_n$.
In other words, we start in an initial state with $\bot$ on the stack
and finish in a final state with the empty stack, and for each
intermediate step there exists a transition.
An input string~$w \in \Sigma^{\scriptscriptstyle +}$ is
\emph{accepted} by $\AUTOMATON$ if there exists an accepting run
starting in $\langle q, w, \bot\rangle$ with $q \in I$.
The language $L(\AUTOMATON)$ \emph{accepted} by the PDA~$\AUTOMATON$
is the set of accepted input strings and thus given by
\begin{align*}
  L(\AUTOMATON) = \bigcup_{(q, q') \in I \times F} \bigl\{ w \in
  \Sigma^{\scriptscriptstyle +} \mid \langle q, w, \bot\rangle
  \vdash^*_\AUTOMATON \langle q', \varepsilon,
  \varepsilon\rangle \bigr\} \enspace.
\end{align*}
%
%
Note that our PDA are $\varepsilon$"~free (in the sense that each
transition induces moves that process an input symbol) and have
limited stack access: In each move we can pop a symbol, push a symbol,
or ignore the stack.  Note that we explicitly exclude the case, in
which a symbol is popped and another symbol is pushed at the same
time.
However, this restriction has no influence on the expressive power
(see~\cite[Corollary~12]{drodzikui19} for the weighted scenario; an
instantiation of the result with the Boolean semiring yields the
unweighted case).  We also note that no moves are possible anymore
once the stack is empty.

\subsection{Tree Languages}
In this paper, we only deal with binary trees since the derivation
trees of CCGs are binary.  We therefore build trees over \emph{ranked
  sets} $\Sigma = \Sigma_0 \cup \Sigma_1 \cup \Sigma_2$.  If $\Sigma$
is an alphabet, then it is a \emph{ranked alphabet}.  For every $k \in
\{0,\, 1,\, 2\}$, we say that symbol $a \in \Sigma_k$ has
\emph{rank}~$k$.  We write $\TREES{\Sigma_2,\, \Sigma_1}(\Sigma_0)$
for the set of all trees over~$\Sigma$, which is the smallest set~$T$
such that $c(\seq t1k) \in T$ for all $k \in \{0, 1, 2\}$, $c \in
\Sigma_k$, and $\seq t1k \in T$.  As usual, we write just~$a$ for
leaves~$a()$ with $a \in \Sigma_0$.  A \emph{tree language} is a
subset $T \subseteq \TREES{\Sigma_2, \emptyset}(\Sigma_0)$.  Let $T =
\TREES{\Sigma_2, \Sigma_1}(\Sigma_0)$.  The map $\POS \colon T \to
\Pplus \bigl([2]^* \bigr)$ assigns 
Gorn tree addresses~\cite{gorn1965explicit} to a tree, where
$\Pplus(S)$~is the set of all nonempty subsets of~$S$.  Let
\begin{align*}
  \POS \bigl(c(\seq t1k) \bigr)
  &= \{ \varepsilon \}\, \cup \bigcup_{i \in [k]} \{ iw \mid
  w \in \POS(t_i) \}
\end{align*}
for all $k \in \{0, 1, 2\}$, $c \in \Sigma_k$, and $\seq t1k \in T$.
The set of all leaf positions of~$t$ is defined as $\LEAVES(t) =
\bigl\{w \in \POS(t) \mid w1 \notin \POS(t) \bigr\}$.  Given a tree~$t
\in T$ and a position~$w \in \pos(t)$, we write $t \vert_w$~and~$t(w)$
to denote the subtree rooted in~$w$ and the symbol at~$w$,
respectively.  Additionally, we let~$t [ t' ]_w$ be the tree obtained
when replacing the subtree appearing in $t$ at position $w$ by the tree~$t' \in
T$.  Finally, let $\mathord{\yield} \colon T \to
\Sigma_0^{\scriptscriptstyle +}$ be inductively defined by~$\yield(a) = a$ for
all~$a \in \Sigma_0$ and $\yield \bigl(c(\seq t1k) \bigr) =
\yield(t_1) \dotsm \yield(t_k)$ for all $k \in [2]$, $c \in
\Sigma_k$, and $\seq t1k \in T$.
The special leaf symbol~$\SBox$ is reserved and is used to represent a hole in a
tree.  The set~$\CONTEXTS{\Sigma_2,\Sigma_1}(\Sigma_0)$ of contexts
contains all trees of~$\TREES{\Sigma_2,\Sigma_1} \bigl(\Sigma_0 \cup
\{ \SBox \} \bigr)$, in which~$\SBox$ occurs exactly once.  We
write~$\POS_{\SBox}(C)$ to denote the unique position of~$\SBox$ in
the context~$C \in \CONTEXTS{\Sigma_2,\Sigma_1}(\Sigma_0)$.  Moreover,
given $t \in T$ we simply write~$C[t]$ instead of $C[t]_{\POS_{\SBox}(C)}$.
A tuple $(\rho_0, \rho_1, \rho_2)$ is called a \emph{relabeling} if
$\rho_k \colon \Sigma_k \to \Delta_k$ for all $k \in \{0, 1, 2\}$ and
ranked set $\Delta$.
It induces the map $\rho \colon T \to \TREES{\Delta_2,
  \Delta_1}(\Delta_0)$ given by $\rho 
\bigl(c(\seq t1k) \bigr) =  \bigl(\rho_k(c) \bigr) \bigl(\rho(t_1),
\dotsc, \rho(t_k) \bigr)$ for all $k \in \{0, 1, 2\}$, $c \in
\Sigma_k$ and $\seq t1k \in T$.

\subsection{Combinatory Categorial Grammar}
In the following, we give a short introduction to CCG.  Given an
alphabet $A$ of \emph{atoms} or \emph{atomic categories} and a set of
\emph{slashes} $D = \{ \SLASHF, \SLASHB \}$ indicating directionality,
the set of \emph{categories} is defined as $\CATEGORIES{A} = \TREES{D,
  \emptyset}(A)$.  We usually write the categories in infix notation
and the slashes are left-associative by convention, so each category
takes the form $c = a \,\vert_1 \,c_1 \,\cdots \,\vert_k \,c_k$ where
$a \in A$, $\vert_i \in D$, $c_i \in \CATEGORIES{A}$ for all $i \in \{
1, \dots, k \}$.  The atom~$a$ is called the \emph{target} of~$c$ and
written as~$\TARGET(c)$.  The slash-category pairs $\vert_i \,c_i$
are called \emph{arguments} and their number~$k$ is called the
\emph{arity} of~$c$ and denoted by~$\ARITY(c)$.  In addition, we write
$\ARGUMENT(c, i)$ to get the $i$"~th argument~$\vert_i\,c_i$ of~$c$.
In the sense of trees, the sequence of arguments is a context $\SBox
\,\vert_1 \,c_1 \,\cdots \,\vert_k \,c_k$.  The set of \emph{argument
  contexts} is denoted by $\ARGUMENTS A \subseteq
\CONTEXTS{D,\emptyset}(A)$.  We distinguish between two types of
categories.  In \emph{first-order categories}, all arguments are
atomic, whereas in \emph{higher-order categories}, the arguments can
have arguments themselves.

Next, we describe how two neighboring categories can be
combined.  Intuitively, the direction of the slash determines on which
side a category matching the argument is expected.  Hence there are two
types of rules.  Despite the conventions for inference systems,
we put the inputs (premises) below and the output (conclusion) above
to make the shape of the proof tree apparent.  A
\emph{rule of degree $k$} with $k \in \mathbb{N}$ has one of the
following forms:
\begin{align*}
\frac{ax \,|_1 \,c_1 \cdots \,|_k \,c_k}{
ax \SLASHF c \qquad c \,|_1 \,c_1 \,\cdots \,|_k \,c_k} 
&& \text{(forward rule)} \\*[1ex]
\frac{ax \,|_1 \,c_1 \cdots \,|_k \,c_k}{
c \,|_1 \,c_1 \,\cdots \,|_k \,c_k \qquad ax \SLASHB c}
&& \text{(backward rule)}
\end{align*}
where $a \in A$, $c \in \CATEGORIES{A} \cup \{ y \}$, $\vert_i \in D$,
and $c_i \in \CATEGORIES{A} \cup \{ y_i \}$ for all $i \in [k]$.
Here, $y, y_1, \dots, y_k$ are category variables that can match any
category in $\CATEGORIES{A}$ and $x$ is an argument context variable
that can match any argument context in $\ARGUMENTS{A}$.
The category taking the argument ($ax \,\vert \,c$ with $\vert \in D$)
is called \emph{primary category}, the one providing it ($c \,|_1
\,c_1 \cdots \,|_k \,c_k$) is called \emph{secondary category}, and
they are combined to an \emph{output category} ($ax \,|_1 \,c_1
\,\cdots \,|_k \,c_k$).  Given rule $r$, we write $\SECONDARY(r)$ to
refer to the secondary category.  Rules of degree~$0$ will be referred
to as \emph{application rules}, while rules of higher degree are
\emph{composition rules}.  We write $\RULESET(A)$ for the set of all
rules over $A$.  A \emph{rule system} is a pair $\Pi = (A, R)$, where
$A$ is an alphabet and $R \subseteq \RULESET(A)$ is a finite set of
rules over $A$.  Given a rule $r \in R$, we obtain a \emph{ground
  instance} of it by replacing the variables $\{ y, y_1, \dots \}$ by
concrete categories and the variable $x$ by a concrete argument
context.  The ground instances of~$\Pi$ induce a
relation~$\mathord{\to_{\Pi}} \subseteq \CATEGORIES{A}^2 \times
\CATEGORIES{A}$ and we write $\tfrac{c''}{c \enspace
  c'}{\scriptscriptstyle \Pi}$ instead of~$(c, c') \to_\Pi c''$.  The
relation~$\to_\Pi$ extends to a relation $\mathord{\DERIVES_{\Pi}}
\subset (\CATEGORIES{A}^*)^2$ on sequences of categories. It is given
by 
\begin{align*}
  \mathord{\DERIVES_{\Pi}} = \bigcup_{\varphi, \psi \in
  \CATEGORIES{A}^*} \bigg\{ (\varphi c c' \psi,\, \varphi c'' \psi)
  \,\bigg|\,  \frac{c''}{c \enspace c'} {\scriptstyle{\Pi}} \bigg\}
  \enspace .
\end{align*}

A \emph{combinatory categorial grammar} (CCG) is a tuple $\mathcal G
= (\Sigma, A, R, I, \LEXICON)$ that consists of an alphabet $\Sigma$
of \emph{input symbols}, a rule system $(A, R)$, a set $I \subseteq
A$ of \emph{initial categories}, and a finite relation $\LEXICON
\subseteq \Sigma \times \CATEGORIES{A}$ called \emph{lexicon}.  It
is called \emph{$k$"~CCG} if each rule $r \in R$ has degree at
most~$k$, where $k \in \nat$. 

The CCG $\mathcal G$ \emph{generates} the category sequences
$\CATEGORYSEQ \subseteq \CATEGORIES{A}^*$ and the string language
$\STRINGLANG(\mathcal G) \subseteq \Sigma^*$ given by
\begin{align*}
  \CATEGORYSEQ = \bigcup_{a_0 \in I} \big\{\varphi \in
  \CATEGORIES{A}^* \,\big|\, \varphi \DERIVES_{(A,R)}^* a_0 \big\} 
\end{align*}
and $\STRINGLANG(\mathcal G) = \LEXICON^{-1}(\CATEGORYSEQ)$,
where the string language $L(\mathcal G)$ contains all strings
that can be relabeled via the lexicon to a category sequence
in~$\CATEGORYSEQ$.  A tree $t \in
\TREES{\CATEGORIES{A},\emptyset}({\LEXICON}(\Sigma))$ is called
\emph{derivation tree of $\mathcal G$} if $\tfrac{t(w)}{t(w1)
  \quad t(w2)} {\scriptstyle(A,R)}$ for every $w \in \POS(t)
\setminus \LEAVES(t)$.
We denote the set of all derivation trees of~$\mathcal G$ by
$\DEVTREES(\mathcal G)$.

A \emph{category relabeling} $\rho \colon \CATEGORIES{A} \to \Delta$
is a relabeling such that $\rho(c) = \rho(c')$ for all categories $c,
c' \in \CATEGORIES{A}$ with $\TARGET(c) = \TARGET(c')$ and $\ARGUMENT 
\bigl(c,\,\ARITY(c) \bigr) = \ARGUMENT \bigl(c',\,\ARITY(c') \bigr)$.
The relabeled derivation trees $\TREELANG_\rho(\mathcal G) \subseteq
\TREES{\Delta_2,\, \emptyset}(\Delta_0)$ are given by
\begin{align*}
  \TREELANG_{\rho}(\mathcal G) = \big\{ \rho(t) \,\big|\, t \in
  \DEVTREES(\mathcal G),\, t(\varepsilon) \in I \, \big\} \enspace. 
\end{align*}
A tree language $\TREELANG \subseteq \TREES{\Delta_2,\,
  \emptyset}(\Delta_0)$ is \emph{generatable} by $\mathcal G$ if there
is a category relabeling $\rho' \colon \CATEGORIES{A} \to \Delta$
such that $\TREELANG = \TREELANG_{\rho'}(\mathcal G)$.

\begin{example}
  \label{ex:ccg}
  Let $\mathcal G = (\Sigma, A, \RULESET(A, 2), \{ \bot \},
  \LEXICON)$ with $\Sigma = \{\alpha, \beta, \gamma, \delta \}$
  and atoms $A = \{ \bot, a, b, c, d, e \} $ be a CCG with the lexicon
  $\LEXICON$ given below, where $\RULESET(A,2)$~is the set of all
  rules over~$A$ up to degree~$2$.  Thus, it is a $2$"~CCG.
  \begin{align*}
    \LEXICON(\alpha)
     &= \{ a,\,b \} 
     &\LEXICON(\gamma)
     &= \{ d \SLASHF c,\, c \SLASHB a \SLASHF c \} \\*
     \LEXICON(\beta)
     &= \{ c \SLASHB b, \, c \SLASHB b \SLASHB e, \, e,\, e \SLASHF e
       \} 
     &\LEXICON(\delta)
     &= \{ \bot \SLASHF d \}
  \end{align*}
  $\tfrac{\bot x \SLASHBS a \SLASHF c}{\bot x \SLASHF c \quad c
    \SLASHBS a \SLASHF c}$ is a forward rule of degree $2$ in
  $\RULESET(A, 2)$, where $x$~is an argument context and can thus be
  replaced by an arbitrary sequence of arguments.  Utilizing $x =
  \SBox \SLASHB a$ yields the ground instance $\tfrac{\bot \SLASHBS a
    \SLASHBS a \SLASHF c}{\bot \SLASHBS a \SLASHF c \quad c \SLASHBS a
    \SLASHF c}$, which has primary category $c_1 = \bot \SLASHB a
  \SLASHF c$ and secondary category $c_2 = c \SLASHB a \SLASHF c$.
  The latter has target $\TARGET(c_2) = c$ and the two arguments
  $\SLASHB a$ and $\SLASHF c$, so its arity is $\ARITY(c_2) = 2$.

  A derivation tree of~$\mathcal G$ is depicted in
  Figure~\ref{fig:exDer}.  We start at the bottom with categories
  taken from the lexicon in accordance with the input symbols.  Then
  neighboring categories are combined until we arrive at the root with
  initial category~$\bot$, so the input word is accepted. 
\end{example}

\begin{figure}
  \centering
  \raisebox{\depth}{\scalebox{1}[-1]{
      $\INFER{\ROT{\bot}}{%
        \PROJECT[5]{\ROT{a}}{\ROT{\alpha}}
        &
        \INFER{\ROT{\bot \SLASHB a}}{%
          \PROJECT[4]{\ROT{b}}{\ROT{\alpha}}
          &
          \INFER{\ROT{\bot \SLASHB a \SLASHB b}}{%
            \INFER{\ROT{\bot \SLASHB a \SLASHF c}}{
              \INFER{\ROT{\bot \SLASHF c}}{
                \PROJECT[1]{\ROT{\bot \SLASHF d}}{\ROT{\delta}}
                &
                \PROJECT[1]{\ROT{d \SLASHF c}}{\ROT{\gamma}}
              }
              &
              \PROJECT[2]{\ROT{c \SLASHB a \SLASHF c}}{\ROT{\gamma}}
            }
            &
            \INFER{\ROT{c \SLASHB b}}{
              \PROJECT[2]{\ROT{e}}{\ROT{\beta}}
              &
              \PROJECT[2]{\ROT{c \SLASHB b \SLASHB e}}{\ROT{\beta}}
            }
          }
        }
      }$
    }}
  \caption{CCG derivation tree (see Example~\protect{\ref{ex:ccg}})}
  \label{fig:exDer}
\end{figure}

\section{Moore Push-down Automata}
We start by introducing a Moore variant of push-down automata
\cite{autberboa97} that is geared towards our needs and still accepts
the context-free languages (of strings of length~$\geq 2$).  It will
be similar to the push-down Moore machines of Decker, Leucker, and
Thoma~\cite{decleutho13}.  Instead of processing input symbols as part
of transitions (as in Mealy machines), Moore machines output a unique
input symbol in each state~\cite{fleischner1977equivalence}. Recall
that for every set~$\Gamma$ we have~$\Gamma^{\leq 1} = \{\varepsilon\}
\cup \Gamma$ and additionally let $\Gamma^{\geq 2} = \{w \in \Gamma^*
\mid 2 \leq \abs w\}$ be the strings of length  at least~$2$. 

\begin{definition}
  A \emph{Moore push-down automaton} (MPDA) is defined as a tuple 
  $\AUTOMATON = (Q, \Sigma, \Gamma, \delta, \tau, I, F)$ that consists
  of
  \begin{enumerate}[label=(\roman*)]
  \item finite sets~$Q$, $\Sigma$, and~$\Gamma$ of \emph{states},
  \emph{input symbols}, and \emph{stack symbols}, respectively,
  \item a set $\delta \subseteq \bigl(Q \times \Gamma^{\leq 1} \times
  \Gamma^{\leq 1} \times Q \bigr) \setminus \bigl(Q \times \Gamma
  \times \Gamma \times Q \bigr)$ of \emph{transitions},
  \item an output function~$\tau \colon Q \to \Sigma$, and
  \item sets~$I, F \subseteq Q$ of \emph{initial} and \emph{final
      states}, respectively.
  \end{enumerate}
\end{definition}

Due to the definition of~$\delta$, as for the PDA of
Definition~\ref{def:pda}, in a single step we can either push or pop a
single stack symbol or ignore the stack.
In the following, let $\AUTOMATON = (Q, \Sigma, \Gamma, \delta, \tau,
I, F)$ be an MPDA.  On the set~$\textrm{Conf}_\AUTOMATON = Q \times
\Gamma^*$ of configurations of~$\AUTOMATON$ the \emph{move relation}
$\mathord{\vdash_\AUTOMATON} \subseteq \textrm{Conf}_\AUTOMATON^2$ is
\begin{align*}
  \mathord{\vdash_\AUTOMATON} = \hspace{-.5em} \bigcup_{\substack{(q,
  \gamma, \gamma', q') \in \delta \\ \alpha \in \Gamma^*}}
  \hspace{-.5em} \Bigl\{ \bigl( \langle q, \gamma\alpha\rangle,\,
  \langle q',  \gamma'\alpha\rangle \bigr) \in
  \textrm{Conf}_\AUTOMATON^2 \;\Big|\; \gamma\alpha \neq \varepsilon
  \Bigr\}
\end{align*}
and a configuration $\langle q, \alpha\rangle \in
\textrm{Conf}_\AUTOMATON$ is \emph{initial} (respectively,
\emph{final}) if $q \in I$ and $\alpha \in \Gamma$ (respectively, $q
\in F$ and $\alpha = \varepsilon$).
An \emph{accepting run} is defined in the same manner as for PDA.
However, note that contrary to PDA we can start with an arbitrary
symbol on the stack.
The language $L(\AUTOMATON)$ \emph{accepted} by~$\AUTOMATON$
contains exactly those strings~$w \in \Sigma^*$, for which there
exists an accepting run $\langle q_0, \alpha_0\rangle, \dotsc, \langle
q_n, \alpha_n\rangle$ such that $w = \tau(q_0) \dotsm \tau(q_n)$.
Thus, we accept the strings that are output symbol-by-symbol by the
states attained during an accepting run.  As usual, two MPDA are
\emph{equivalent} if they accept the same language.  Since no initial
configuration is final, each accepting run has length at least~$2$, so
we can only accept strings of length at least~$2$.  While we could
adjust the model to remove this restriction, the presented version
serves our later purposes best.

\begin{theorem}
  \label{lm:soPDA}
  MPDA accept the context-free languages of strings of length at
  least~$2$.
\end{theorem}

\begin{proof}
  The straightforward part of the proof is to show that each language
  accepted by an MPDA is context-free.  For the converse, let $L
  \subseteq \Sigma^{\geq 2}$ be context-free and
  $\AUTOMATON = (Q, \Sigma, \Gamma, \delta, I, \bot, F)$ be a PDA such
  that $L(\AUTOMATON) = L$.
  We assume without loss of generality that the initial states of
  $\AUTOMATON$ have no incoming transitions and that the final states
  of~$\AUTOMATON$ have no outgoing transitions and all their incoming
  transitions pop $\bot$.
  We will construct an  MPDA~$\AUTOMATON'$ with
  $L(\AUTOMATON') = L$ in the spirit of the classical conversion from
  Mealy to Moore machines.  The main idea is to shift the input
  symbol~$a$ from the transition~$(q, a, \gamma, \gamma', q') \in
  \delta$ to the target state~$q'$.  Additionally, since there is
  always one more configuration compared to the number of moves (and
  thus involved transitions) in an accepting run, the first move
  needs to be avoided in~$\AUTOMATON'$.
  If the corresponding transition pushes a symbol to the stack,
  we have to store it in the target state of the transition.
  This state becomes an initial state of $\AUTOMATON'$.
  To be able to discern if the stored symbol can be deleted,
  $\AUTOMATON'$ needs to be aware whether the stack currently contains
  only one symbol, since in $\AUTOMATON$ the symbol pushed in the
  first transition lies above the bottom symbol.
  Since we clearly cannot store the size of the
  current stack in the state, we need to mark the symbol at the bottom
  of the stack.

  Formally, we construct the MPDA $\AUTOMATON' = (Q', \Sigma, \Gamma',
  \delta', \pi_2, I', F')$ with
  \begin{itemize}[noitemsep]
  \item $Q' = Q \times \Sigma \times \Gamma^{\leq 1} \times \{0, 1\}$,
  \item $I' = \bigl\{ \langle q', a, \gamma', 1\rangle \in Q' \mid
    \exists q \in I \colon (q, a, \gamma, \gamma', q') \in \delta
    \bigr\}$,
  \item $\Gamma' = \Gamma_\bot \times \{0, 1\}$,
  \item $F' = F \times \Sigma \times \{\varepsilon\} \times \{1\}$,
  \item and the transitions
  \end{itemize}
  \begin{align}
    \delta' = \bigcup_{\substack{a' \in \Sigma,\, b \in \{0, 1\} \\
    (q, a, \gamma, \gamma', q') \in \delta \\
    \gamma'' \in \Gamma_\bot^{\leq 1} }} \biggl(
    &\Bigl\{ \bigl(\langle q, a', \gamma'', b \rangle, \varepsilon,
      \varepsilon, \langle q', a, \gamma'', b \rangle \bigr)
      \,\Big|\, \gamma = \gamma' = \varepsilon \Bigr\}\
      \cup \label{eq:ignore_op}
    \\*[-6ex]
    &{} \Bigl\{ \bigl(\langle q, a', \gamma'', b \rangle, \varepsilon,
      \langle \gamma', b\rangle, \langle q', a, \gamma'', 0 \rangle
      \bigr) \,\Big|\, \gamma = \varepsilon, \gamma' \neq \varepsilon
      \Bigr\}\ \cup \label{eq:push_op} \\*
    &{} \Bigl\{ \bigl(\langle q, a', \gamma'', 0 \rangle, \langle
      \gamma, b \rangle, \varepsilon, \langle q', a, \gamma'', b
      \rangle \bigr)\,\Big|\, \gamma \neq \varepsilon, \gamma' =
      \varepsilon\Bigr\}\ \cup \label{eq:pop_op1} \\*
    &{} \Bigl\{ \bigl(\langle q, a', \gamma, 1 \rangle,
      \varepsilon, \varepsilon, \langle q', a, \varepsilon, 1 \rangle
      \bigr) \,\Big|\, \gamma \neq \varepsilon, \gamma' = \varepsilon
      \Bigr\}\ \cup  \label{eq:pop_op2} 
    \\*[-0.5ex]
    &{} \Bigl\{ \bigl(\langle q, a', \varepsilon, 1 \rangle, \langle
      \gamma, 1 \rangle, \varepsilon, \langle q', a, \varepsilon, 1
      \rangle \bigr) \,\Big|\, \gamma \neq \varepsilon, \gamma' =
      \varepsilon \Bigr\} \label{eq:pop_op3} \biggr) \enspace.
  \end{align} 

  While transitions that ignore the stack~\eqref{eq:ignore_op} or
  push to the stack~\eqref{eq:push_op} can be adopted easily, we have
  three variants \eqref{eq:pop_op1}--\eqref{eq:pop_op3} of transitions
  that pop from the stack.  If we have not reached the bottom of the
  stack yet, then we can pop symbols without
  problems~\eqref{eq:pop_op1}.  However, when only the initial stack
  symbol is left, we first have to remove the stored
  symbol~\eqref{eq:pop_op2} before we can pop the initial stack
  symbol~\eqref{eq:pop_op3}.

  Let $w = \word a1n$ with~$\seq a1n \in \Sigma$ be an input string.
  First we assume that the first move of~$\AUTOMATON$ pushes the
  symbol~$\gamma_1$ to the stack, which then gets popped in the
  $i$"-th move.  Any such sequence of configurations
  \begin{align*}
    &\langle q_1,\, \word a1n,\, \bot \rangle
    &\vdash_\AUTOMATON
    &\,\langle q_2,\, \word a2n,\, \gamma_1 \bot \rangle
    &\vdash_\AUTOMATON
    &\,\langle q_3,\, \word a3n,\, \gamma_2 \gamma_1 \bot \rangle
      \vdash_\AUTOMATON \cdots \\*
    \vdash_\AUTOMATON\ 
    &\langle q_i,\, \word ain,\, \gamma_1 \bot \rangle
    &\vdash_\AUTOMATON
    &\,\langle q_{i+1},\, \word a{i+1}n,\, \bot \rangle
    &\vdash_\AUTOMATON
    &\ \cdots \\*
    \vdash_\AUTOMATON\ 
    &\langle q_n,\, a_n,\, \bot\rangle
    &\vdash_\AUTOMATON
    &\,\langle q_{n+1},\, \varepsilon,\, \varepsilon \rangle
  \end{align*}
  in~$\AUTOMATON$ with~$q_1 \in I$ and~$q_{n+1} \in F$ yields the
  corresponding sequence of configurations
  \begin{align*}
    &\langle \langle q_2, a_1, \gamma_1, 1\rangle ,\, \langle \bot,
      1\rangle  \rangle
    &\vdash_{\AUTOMATON'}
    &\ \ \langle \langle q_3, a_2, \gamma_1, 0\rangle ,\, \langle
      \gamma_2, 1\rangle  \langle \bot, 1\rangle  \rangle 
    &\vdash_{\AUTOMATON'}\ 
      \cdots \\*
    \vdash_{\AUTOMATON'}\ \ 
    &\langle \langle q_i, a_{i-1}, \gamma_1, 1\rangle ,\, \langle
      \bot, 1\rangle  \rangle
    &\vdash_{\AUTOMATON'}
    &\ \ \langle \langle q_{i+1}, a_{i}, \varepsilon, 1\rangle ,\,
      \langle \bot, 1\rangle  \rangle 
    &\vdash_{\AUTOMATON'}\ 
      \cdots \\*
    \vdash_{\AUTOMATON'}\ \ 
    &\langle \langle q_n, a_{n-1}, \varepsilon, 1\rangle ,\, \langle
      \bot, 1\rangle  \rangle 
    &\vdash_{\AUTOMATON'}
    &\ \ \langle \langle q_{n+1}, a_{n}, \varepsilon, 1\rangle ,\,
      \varepsilon \rangle 
    &
  \end{align*}
  which is an accepting run of~$\AUTOMATON'$ and vice versa.

  On the other hand, if the first move of~$\AUTOMATON$ ignores the
  stack, then $\AUTOMATON'$~simulates the moves starting in
  configuration~$\langle \langle q_2, a_1, \varepsilon, 1 \rangle,
  \langle \bot, 1 \rangle \rangle$.  Thus, each string of length at
  least~$2$ accepted by~$\AUTOMATON$ is also accepted by~$\AUTOMATON'$
  and vice versa, which proves $L(\mathcal A') = L = L(\mathcal A)$.
\end{proof}

The MPDA~$\AUTOMATON$ is \emph{pop-normalized} if there exists a
map~$\mathord{\RETURN} \colon \Gamma \to Q$ such that $q' =
\RETURN(\gamma)$ for every transition~$(q, \gamma, \varepsilon, q')
\in \delta$.  In other words, for each stack symbol~$\gamma \in
\Gamma$ there is a unique state~$\RETURN(\gamma)$ that the MPDA enters
whenever $\gamma$~is popped from the stack.
 
Later on, we will simulate the runs of an MPDA in a CCG such that
subsequent configurations are represented by subsequent primary
categories.  Pop transitions are modeled by removing the last argument
of a category. Thus, the target state has to be stored in the previous
argument.  This argument is added when the according push transition
is simulated, so at that point we already have to be aware in which
state the MPDA will end up after popping the symbol again.  This will
be explained in more detail in Section~\ref{sec:constructingccg}.

We can easily establish this property by storing a state in each stack
symbol. Each push transition is replaced by one variant for each state
(i.e.,~we guess a state when pushing), but when a symbol is popped, this
is only allowed if the state stored in it coincides with the
target state.

\begin{lemma}
  \label{lm:returnstate}
  For every MPDA we can construct an equivalent pop-normalized MPDA.
\end{lemma}

\begin{proof}
  Given an MPDA $\AUTOMATON = (Q, \Sigma, \Gamma, \delta, \tau, I, F)$, 
  we extend each stack symbol by a state and let $\Gamma' = \Gamma
  \times Q$ as well as $\mathord{\RETURN} = \pi_2$,
  i.e., $\RETURN \bigl(\langle \gamma, q\rangle \bigr) = q$ for all
  $\langle \gamma, q\rangle \in \Gamma'$.  All transitions that push a
  symbol to the stack also guess the state that is entered when
  that symbol is eventually popped.  Hence we construct~$\AUTOMATON' =
  (Q, \Sigma, \Gamma', \delta', \tau, I, F)$ with
  \begin{align*}
    \delta' = \bigcup_{(q, \gamma, \gamma', q') \in \delta} {}
    &\Biggl( \Bigl\{ (q, \varepsilon, \varepsilon, q') \;\Big|\;
      \gamma = \gamma' = \varepsilon \Bigr\} \cup \Bigl\{ \bigl(q,
      \langle \gamma, q' \rangle, \varepsilon, q' \bigr) \;\Big|\;
      \gamma' = \varepsilon \Bigr\} \cup {} \\*
    &\phantom{\Biggl(} \Bigl\{ \bigl(q, \varepsilon, \langle \gamma',
      q''\rangle, q' \bigr) \;\Big|\; \gamma = \varepsilon,\, q'' \in
      Q \Bigr\} \Biggr)
      \enspace.
  \end{align*} 
  It is obvious that for every accepting run of~$\AUTOMATON$ there is
  an accepting run of~$\AUTOMATON'$, in which all the guesses were
  correct.  Note that $\AUTOMATON'$~starts with an arbitrary symbol on
  the stack, so we can find a run where the second component of this
  symbol coincides with the final state that is reached by popping
  this symbol.  Similarly, every accepting run of~$\AUTOMATON'$ can be
  translated into an accepting run of~$\AUTOMATON$ by projecting each
  stack symbol to its first component.  Hence
  $\AUTOMATON$~and~$\AUTOMATON'$ are equivalent and $\AUTOMATON'$ is
  clearly pop-normalized.
\end{proof}

The next statement shows that we can provide a form of look-ahead on
the output.  In each new symbol we store the current as well as the
next output symbol. 
We will briefly sketch why this look-ahead is necessary.
Before constructing the CCG, the MPDA will be used to model a spine
grammar.  The next output symbol of the MPDA corresponds to the label
of the parent node along a ``spine'' of a tree generated by the
spine grammar.  From this parent node we can determine the
label of its other child.  This information will be used in the CCG
to control which secondary categories are allowed as neighboring
combination partners.

\begin{lemma}
  \label{lm:succ}
  Let $L \subseteq \Sigma^*$ with $\mathord{\triangleleft} \notin
  \Sigma$ be a context-free language.  Then the language~$\Next(L)$ is
  context-free as well, where
  \[ \Next(L) = 
  \bigcup_{n \in \nat}
    \bigl\{ \langle\sigma_2, \sigma_1\rangle \cdots \langle \sigma_n,
    \sigma_{n-1}\rangle \langle \triangleleft, \sigma_n \rangle
    \;\big|\; \seq\sigma 1n \in \Sigma,\, \word \sigma1n \in L
    \bigr\} \enspace . \]
\end{lemma}

\begin{proof}
  We define $\Sigma' = \Sigma\, \cup\, \{\mathord{\triangleleft}\}$
  and $\Sigma'' = \Sigma' \times \Sigma$.   Let us consider the
  homomorphism~$\pi_2 \colon (\Sigma'')^* \to \Sigma^*$ given by
  $\pi_2 \bigl(\langle \sigma', \sigma\rangle \bigr) = \sigma$ for
  all~$\langle \sigma', \sigma \rangle \in \Sigma''$.  Since the 
  context-free languages are closed under inverse
  homomorphisms~\cite[Chapter~2, Theorem~2.1]{ber79}, the
  language~$L'' = \pi^{-1}(L)$ is also context-free.  Additionally,
  the language 
  \[ R = \bigcup_{n \in \nat} \bigl\{ \langle\sigma_2, \sigma_1\rangle
    \langle\sigma_3, \sigma_2\rangle \cdots \langle \sigma_n,
    \sigma_{n-1}\rangle \langle \triangleleft, \sigma_n \rangle
    \;\big|\; \seq\sigma 1n \in \Sigma \bigr\} \enspace, \] 
  is regular because it is recognized by the NFA~$\AUTOMATON =
  \bigl(\Sigma', \Sigma'', \delta, \Sigma',
  \{\mathord{\triangleleft}\})$ with
  state set and initial states~$\Sigma'$, input alphabet~$\Sigma''$,
  final states~$\{\triangleleft\}$, and transitions
  \[ \delta = \bigl\{ (\sigma, \langle\sigma', \sigma\rangle, \sigma')
    \;\big|\; \sigma \in \Sigma,\, \sigma' \in \Sigma' \bigr\}
    \enspace. \]
  Finally, $\Next(L) = L'' \cap R$ is context-free because the
  intersection of a context-free language with a regular language is
  again context-free~\cite[Chapter~2, Theorem~2.1]{ber79}.
\end{proof}

\begin{corollary}
  \label{cor:MPDA}
  For every context-free language $L \subseteq \Sigma^{\geq 2}$ there
  exists a pop-normalized MPDA~$\mathcal A$ such that~$L(\mathcal A) =
  \Next(L)$.
\end{corollary}

\section{Spine Grammars}
Now we move on to representations of tree languages.  We first recall
context-free tree grammars (CFTG)~\cite{rou69}, but only the monadic
simple variant~\cite{kepser2011equivalence}, i.e., all nonterminals
are either nullary or unary and productions are linear and
nondeleting.

\begin{definition}
  A \emph{simple monadic context-free tree grammar} (sCFTG) is a tuple
  $\mathcal G = (N, \Sigma, S, P)$ consisting of 
  \begin{enumerate}[label=(\roman*)]
  \item disjoint ranked alphabets $N$ and $\Sigma$ of
    \emph{nonterminal} and \emph{terminal symbols} with $N = N_1 \cup
    N_0$ and $\Sigma_1 = \emptyset$,
  \item a nullary \emph{start nonterminal} $S \in N_0$, and
  \item a finite set $P \subseteq P_0 \cup P_1$ of \emph{productions},
    where $P_0 = N_0 \times T_{\Sigma_2, N_1} (N_0 \cup \Sigma_0)$ and
    $P_1 = N_1 \times C_{\Sigma_2, N_1} (N_0 \cup \Sigma_0)$.
  \end{enumerate}
\end{definition}

In the following let $\mathcal G = (N, \Sigma, S, P)$ be an sCFTG.  We
write~$(n, r) \in P$ simply as~$n \to r$.  Given~$t, u \in
T_{\Sigma_2, N_1} (\Sigma_0 \cup N_0)$ we let~$t \Rightarrow_{\mathcal
  G} u$ if there exist $(n \to r) \in P$ and a position
$w \in \pos(t)$ such that (i)~$t|_w = n$ and $u = t[r]_w$ with~$n \in
N_0$, or (ii)~$t|_w = n(t')$ and $u = t 
\bigl[r[t'] \bigr]_w$ with~$n \in N_1$ and~$t' \in T_{\Sigma_2, N_1}
(\Sigma_0 \cup N_0)$.  The \emph{tree language~$T(\mathcal G)$
  generated by~$\mathcal G$} is
\[ T(\mathcal G) = \{t \in T_{\Sigma_2,\emptyset}(\Sigma_0) \mid S
  \Rightarrow_{\mathcal G}^* t \} \enspace. \]
The sCFTG $\mathcal G'$ is \emph{strongly equivalent}
to~$\mathcal G$ if $T(\mathcal G) = T(\mathcal G')$, and it is
\emph{weakly equivalent} to~$\mathcal G$ if $\yield \bigl(T(\mathcal
G) \bigr) = \yield \bigl(T(\mathcal G') \bigr)$.

Spine grammars~\cite{fujiyoshi2000spinal} were originally defined
as a restriction of CFTG and possess the same expressive power as
sCFTG, which follows from the normal form for spine grammars.
Although sCFTG are more established, we elect to utilize spine
grammars because of their essential notion of spines and use a variant
of their normal form.
Deviating from the original
definition~\cite[Definition~3.2]{fujiyoshi2000spinal},
we treat spine grammars as a restriction on sCFTG and equip the
terminal symbols with a ``spine direction'' (instead of the
nonterminals, which is not useful in sCFTG).\footnote{In the original
  definition, productions are not necessarily linear or nondeleting,
and the nonterminals may have rank greater than~$1$.  Nonterminals are
equipped with a \emph{head} that specifies the direction where the
spine continues.  The spine of the right"-hand side of a production
is the path from the root to the unique appearence of the
variable that is in head direction of the (non"-nullary) nonterminal
on the left"-hand side.  All other variables on the left"-hand side of
productions have to appear as children of spinal nodes on the
right"-hand side if they appear at all.}
By creating copies of binary terminal symbols it can be shown that
both variants are equivalent modulo relabeling.  More specifically,
under our definition, each spine grammar is clearly itself an sCFTG
and for each sCFTG~$\mathcal G$ there exist a spine grammar~$\mathcal
G'$ and a relabeling~$\rho$ such that $T(\mathcal G) = \{\rho(t) \mid
t \in T(\mathcal G')\}$.

\begin{definition}
  Let $\mathcal G  = (N, \Sigma, S, P)$ be an sCFTG. It is called
  \emph{spine grammar} if there exists a map~$d \colon \Sigma_2 \to
  \{1, 2\}$ such that $wi \in \Pref(\pos_{\SBox}(C))$ with $i =
  d(C(w))$ for every production~$(n \to C) \in P$ with~$n \in N_1$ and
  $w \in \Pref(\pos_{\SBox}(C))$ with~$C(w) \in \Sigma_2$.
\end{definition}

Henceforth let $\mathcal G  = (N, \Sigma, S, P)$ be a spine grammar
with map~$d \colon \Sigma_2 \to \{1,2\}$.  Consider a production $(n
\to C) \in P$ with~$n \in N_1$.  The \emph{spine} of~$C$ is simply the
path from the root of~$C$ to the unique occurrence~$\pos_{\SBox}(C)$
of~$\SBox$.  The special feature of a spine grammar is that the
symbols along the spine indicate exactly in which direction the spine
continues.  Since only the binary terminal symbols offer branching,
the special feature of spine grammars is the existence of a map~$d$
that tells us for each binary terminal symbol~$\sigma \in \Sigma_2$
whether the spine continues to the left, in which case~$d(\sigma) =
1$, or to the right, in which case~$d(\sigma) = 2$.  This map~$d$,
called \emph{spine direction}, applies to all instances of~$\sigma$ in
all productions with spines.
We will use the term spine also to refer to the paths that follow
the spine direction in a tree generated by a spine grammar.
In this manner, each such tree can be decomposed into a set of spines.

\begin{definition}
  \label{def:normalform}
  Spine grammar~$\mathcal G$ is in \emph{normal form} if each
  $(n \to r) \in P$ is of one of the forms
  \begin{enumerate}[label=(\roman*)]
  \item \emph{start:} $r = b(\alpha)$ or $r = \alpha$ for some $b \in
    \NONTERMINALS_1$ and $\alpha \in \Sigma_0$,
  \item \emph{chain:} $r = b_1 \bigl(b_2(\SBox)\bigr)$ for some $b_1,
    b_2 \in \NONTERMINALS_1$, or
  \item \emph{terminal:} $r = \sigma(\SBox, a)$~or~$r = \sigma(a,
    \SBox)$ for some $\sigma \in \TERMINALS_2$~and~$a \in
    \NONTERMINALS_0 \setminus \{ S \}$.
  \end{enumerate}
\end{definition}

In spine grammars in normal form the start nonterminal is isolated
and cannot occur on the right-hand sides.  The three production types
of the normal form are illustrated in Figure~\ref{fig:normalform}.

Using a single start production followed by a number of chain and
terminal productions, a nullary nonterminal $n$ can be rewritten to a
tree $t$ that consists of a spine of terminals, where each non-spinal
child is a nullary nonterminal.
Formally, for every nullary nonterminal~$n \in N_0$ let 
\[ I_{\mathcal G}(n) = \{t \in T_{\Sigma_2,\, \emptyset}(\Sigma_0 \cup
  N_0) \mid n \mathrel{(\mathord{\Rightarrow_{\mathcal G}} \mathbin;
    \mathord{\Rightarrow_{\mathcal G'}^*})} t\} \] 
where $\mathcal G'$ is the spine grammar~$\mathcal G$ without start
productions; i.e., $\mathcal G' = (N, \Sigma, S, P')$ with productions
$P' = \{(n \to r) \in P \mid n \in N_1\}$.
So we perform a single derivation step using the productions
of~$\mathcal G$ followed by any number of derivation steps using only
productions of $\mathcal G'$.  The elements of~$I_{\mathcal G}(n)$ are
called \emph{spinal trees} for~$n$ and their \emph{spine generator}
is~$n$.  By a suitable renaming of nonterminals we can always achieve
that the spine generator does not occur in any of its spinal trees.
Accordingly, the spine grammar $\mathcal G$ is \emph{normalized} if it
is in normal form and $I_{\mathcal G}(n) \subseteq T_{\Sigma_2,\,
  \emptyset}(\Sigma_0 \cup (N_0 \setminus \{n\}))$ for every nullary
nonterminal $n \in N_0$.

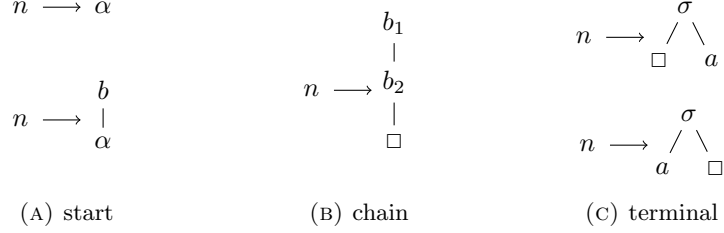
\begin{figure}
  \centering
  \hspace{3em}
  \begin{subfigure}[b]{0.2\textwidth}
    \centering
    \begin{tikzpicture}
      \node at (0.2,0) {$n$}[level distance=2em];
      \draw[->] (0.5,0) -- (1,0);
      \node at (1.3,0) {$\alpha$};
      \node at (0.2,-1.5) {$n$}[level distance=2em];
      \draw[->] (0.5,-1.5) -- (1,-1.5);
      \node at (1.3,-1.1) {$b$} [grow'=down,level distance=2em]
      child {node {$\alpha$} };
    \end{tikzpicture}
    \vspace{1em}
    \caption{start}
  \end{subfigure}
  \hfill
  \begin{subfigure}[b]{0.2\textwidth}
    \centering
    \begin{tikzpicture}
      \node at (0.2,-0.9) {$n$}[level distance=2em];
      \draw[->] (0.5,-0.9) -- (1,-0.9);
      \node at (1.3,0) {$b_1$} [grow'=down,level distance=2.25em]
      child {node {$b_2$}
        child {node {$\SBox$}} };
    \end{tikzpicture}
    \vspace{1em}
    \caption{chain}
  \end{subfigure}
  \hfill
  \begin{subfigure}[b]{0.2\textwidth}
    \centering
    \begin{tikzpicture}
      \node at (0.2,-0.9) {$n$}[level distance=2em];
      \draw[->] (0.5,-0.9) -- (1,-0.9);
      \node at (1.55,-0.5) {$\sigma$}
        [grow'=down,level distance=2em,sibling distance=2em]
      child {node {$a$}}
      child {node {$\SBox$}};
    \end{tikzpicture}
    \bigskip
  
    \begin{tikzpicture}
      \node at (0.2,-0.9) {$n$}[level distance=2em];
      \draw[->] (0.5,-0.9) -- (1,-0.9);
      \node at (1.55,-0.5) {$\sigma$}
        [grow'=down,level distance=2em,sibling distance=2em]
      child {node {$\SBox$}}
      child {node {$a$}};
    \end{tikzpicture}
    \caption{terminal}
  \end{subfigure}
  \hspace{3em}
  \caption{Types of productions of spine grammars in normal form
    (see Definition~\protect{\ref{def:normalform}})}
  \label{fig:normalform}
\end{figure}

\begin{example}
  \label{ex:spinegrammar}
  Let spine grammar $\mathcal G = (N, \TERMINALS, s, P)$
  with $N_1 = \{t, a, b, c, b', e\}$, $N_0 = \{ s, \overline{a},
  \overline{b}, \overline{c},\overline{e} \}$, $\Sigma_2 = \{
  \alpha_2, \beta_2, \gamma_2, \eta_2 \} $, $\Sigma_0 = \{ \alpha,
  \beta, \gamma, \delta \}$, and $P$~as shown below.
  \begin{align*}
    & &
    \overline{a}
    &\to \alpha
    & t
    &\to a \bigl(b'(\SBox) \bigr)
    & a
    &\to \alpha_2(\overline{a}, \SBox) \\*
    s
    &\to t(\delta)&
    \overline{b}
    &\to \beta
    & b'
    &\to b \bigl(c(\SBox) \bigr)
    & b
    &\to \beta_2(\SBox, \overline{b}) \\*
    \overline{b}
    & \to e(\beta)&
    \overline{c}
    &\to \gamma
    & b
    &\to a \bigl(b'(\SBox) \bigr)
    & c
    &\to \gamma_2(\SBox, \overline{c}) \\*
    & &
    \overline{e}
    &\to \beta
    & e
    &\to e \bigl(e(\SBox) \bigr)
    & e
    &\to \eta_2(\overline{e}, \SBox)
  \end{align*}
  The tree in Figure~\ref{fig:genTree}, in which the spines are marked
  by thick edges, is generated by~$\mathcal G$.  The spinal tree
  corresponding to the main spine of the depicted tree is shown in
  Figure~\ref{fig:spinaltree}.  The yield of~$T(\mathcal G)$ is $\{
  \alpha^n\, \delta\, \gamma^n\, \beta^m \mid n, m \in \natp \}$.
\end{example}

The following result is a variant of Theorem~1 of Fujiyoshi and
Kasai~\cite{fujiyoshi2000spinal}.

\begin{theorem}
  \label{thm:normalform}
  For every spine grammar there is a strongly equivalent normalized
  spine grammar.
\end{theorem}

\begin{proof}
  The normal form of Fujiyoshi and
  Kasai~\cite[Definition~4.2]{fujiyoshi2000spinal} is different from
  that of Definition~\ref{def:normalform} in that the first case of
  productions of type~(i) has right"-hand side $r = b(a)$ with $b \in
  N_1$, $a \in N_0$, that productions of type~(ii) have right"-hand
  side $r = b_1(\cdots(b_m(\SBox))\cdots)$ with $m \in \nat$ and $b_i
  \in N_1$ for $i \in [m]$, and that in productions of type~(iii) the
  start nonterminal is not excluded from the set of nonterminals
  that can be produced.
  When starting from their normal form, standard techniques can be
  used to modify the grammar such that all productions $n \to r $ with
  $r = b_1(\cdots(b_m(\SBox))~\cdots)$ have $m=2$,
  that the start nonterminal is isolated, and that
  no nullary nonterminal can derive a spinal tree containing the same
  nonterminal.
  We therefore assume that these conditions are already met.

  Let $\mathcal G = (N,\Sigma,S,P)$ be a spine grammar that is already
  in the desired form except for the set of productions $P_1 = \{ n
  \to b(a) \mid n,a \in N_0,\, b \in N_1\} \subseteq P$.  Let
  $\mathcal G' = (N_0 \cup N_1', \Sigma,  S, P')$ with unary
  nonterminals $N_1' = N_1 \cup N_0 \times \Sigma_0$ and productions
  \begin{align*}
    P' = \, \big( P  \setminus P_1  \big) \,\cup\,
    & \big\{n \to \langle n,\alpha \rangle(\alpha) \mid n \to b(a) \in
      P_1, \alpha \in \Sigma_0 \big\} \,\cup\, \\
    & \big\{ \langle n,\alpha \rangle \to b(\langle a,\alpha
      \rangle(\SBox)) \mid n \to b(a) \in P_1, \alpha \in
      \Sigma_0\big\}\,\cup\, \\
    & \big\{ \langle a,\alpha \rangle \to \SBox \mid a \to \alpha \in
      P \big\} 
      \enspace .
  \end{align*}
  When a nonterminal is expanded to a non"-trivial spine, the terminal
  symbol at the bottom of that spine is guessed.  That symbol is
  immediately produced and stored in its parent nonterminal.  If the
  original nonterminal corresponding to the parent can be replaced by
  the guessed terminal symbol in the original grammar, the parent can
  be removed instead since the terminal symbol was already produced
  (see also~\cite[Section~5]{maletti2012strong} for a similar
  construction).
  It is easy to verify that~$\mathcal G'$ still generates the same
  tree language as~$\mathcal G$.

  After these modifications, the set $P'$ contains collapsing
  productions $a \to \SBox$ with $a \in N_1$ that are not allowed in
  our normal form. They are subsequently removed using the standard
  techniques for removal of $\varepsilon$-productions and unit
  productions from CFG to obtain a spine grammar of the desired normal
  form.
\end{proof}

\tikzmath{\sibldist1 = 3.5;} 
\tikzmath{\sibldist2 = 5;}
\tikzmath{\sibldist3 = 3.5;}

\tikzmath{\sibldist4 = 5;} 
\tikzmath{\sibldist5 = 6.5;}
\tikzmath{\sibldist6 = 3.5;}

\tikzmath{\lvldist1 = 1.9;} 
\tikzmath{\lvldist2 = 2.3;} 

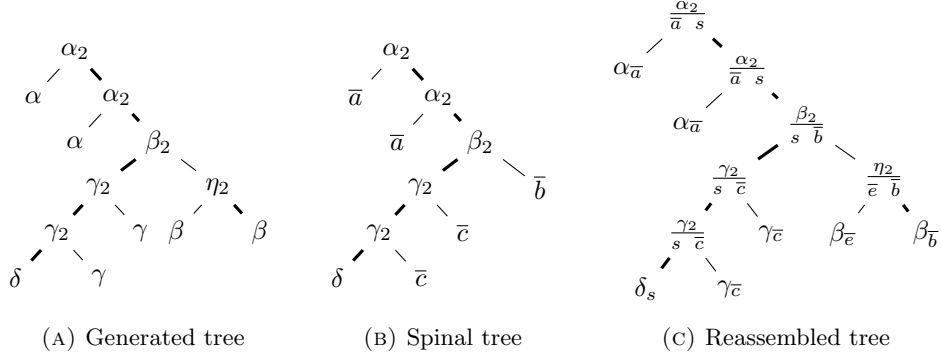
\begin{figure*}
    \begin{subfigure}[b]{0.30\textwidth}
  \centering
  \begin{tikzpicture}[scale=0.9, emph/.style={edge from parent/.style={black,very thick,draw}},norm/.style={edge from parent/.style={black,thin,draw}}]
    \node at (0,0) {$\alpha_2$}
      [grow'=down,level distance=\lvldist1 em,sibling distance=\sibldist1 em]
      child[emph] {node {$\alpha_2$}
      child[emph] {node {$\beta_2$}
        [level distance=\lvldist1 em,sibling distance=\sibldist2 em]
        child[norm] {node {$\eta_2$}
          [level distance=\lvldist1 em,sibling distance=\sibldist3 em]
          child[emph] {node {$\beta$}
          }
          child {node[norm] {$\beta$}}
        }
        child[emph] {node {$\gamma_2$}
          [level distance=\lvldist1 em,sibling distance=\sibldist1 em]
          child[norm] {node {$\gamma$}}
          child[emph] {node {$\gamma_2$}
            child[norm] {node {$\gamma$}}
            child[emph] {node {$\delta$}}
          }
        }
      }
      child[norm] {node {$\alpha$}}
    }
    child[norm] {node {$\alpha$}};
  \end{tikzpicture}
  \vspace{0.5em}
  \caption{Generated tree}
  \label{fig:genTree}
    \end{subfigure}
    \hfill
    \begin{subfigure}[b]{0.30\textwidth}
  \centering
  \begin{tikzpicture}[scale=0.9, emph/.style={edge from parent/.style={black,very thick,draw}},norm/.style={edge from parent/.style={black,thin,draw}}]
    \node at (0,0) {$\alpha_2$}
      [grow'=down,level distance=\lvldist1 em,sibling distance=\sibldist1 em]
      child[emph] {node {$\alpha_2$}
      child[emph] {node {$\beta_2$}
        [level distance=\lvldist1 em,sibling distance=\sibldist2 em]
        child[norm] {node {$\overline{b}$}
        }
        child[emph] {node {$\gamma_2$}
          [level distance=\lvldist1 em,sibling distance=\sibldist3 em]
          child[norm] {node {$\overline{c}$}}
          child[emph] {node {$\gamma_2$}
            child[norm] {node {$\overline{c}$}}
            child[emph] {node {$\delta$}}
          }
        }
      }
      child[norm] {node {$\overline{a}$}}
    }
    child[norm] {node {$\overline{a}$}};
  \end{tikzpicture}
  \vspace{0.5em}
  \caption{Spinal tree} 
  \label{fig:spinaltree}
    \end{subfigure}
    \hfill
    \begin{subfigure}[b]{0.37\textwidth}
        \centering
        \begin{tikzpicture}[scale=0.9, emph/.style={edge from parent/.style={very thick,draw}}
              ,norm/.style={edge from parent/.style={thin,draw}}]
            \node at (0,0) {$\ANNOA$}
            [grow'=down,level distance=\lvldist2 em,sibling distance=\sibldist4 em]
            child[emph] {node {$\ANNOA$}
              child[emph] {node {$\ANNOB$}
            [grow'=down,level distance=\lvldist2 em,sibling distance=\sibldist5 em]
                  child[norm] {node {$\ANNOE$}
            [grow'=down,level distance=\lvldist2 em,sibling distance=\sibldist6 em]
                      child[emph] {node {$\beta_{\overline{b}}$}
                      }
                      child[norm] {node  {$\beta_{\overline{e}}$}
                      }
                  }
                  child[emph] {node {$\ANNOC$}
            [grow'=down,level distance=\lvldist2 em,sibling distance=\sibldist6 em]
                      child[norm] {node  {$\gamma_{\overline{c}}$}
                      }
                      child[emph] {node {$\ANNOC$}
                          child[norm] {node  {$\gamma_{\overline{c}}$}
                          }
                          child[emph] {node {$\delta_s$}
                          }
                      }
                  }
            }
            child[norm] {node  {$\alpha_{\overline{a}}$}
            }
          }
          child[norm] {node  {$\alpha_{\overline{a}}$}
          }
          ;
        \end{tikzpicture}
        \caption{Reassembled tree}
        \label{fig:assembledTree}
    \end{subfigure}
    \caption{Tree generated by spine grammar $\mathcal G$,
    a spinal tree in $I_{\mathcal G}(s)$ (see
    Example~\protect{\ref{ex:spinegrammar}}), and a
    tree in $\mathcal F(\SPINES(\mathcal G))_S$ reassembled from spines (see
    Example~\protect{\ref{ex:reassembly}})}
\end{figure*}

\section{Tree-adjoining Grammars}
Before we proceed, we will briefly introduce TAG and sketch
  how a spine grammar is obtained from it.  TAG is a mildly
  context"-sensitive grammar formalism that operates on a set of
  \emph{elementary trees} of which a subset is \emph{initial}.  To
  generate a tree, we start with an initial tree and successively
  splice elementary trees into nodes using \emph{adjunction}
  operations.  In an adjunction, we select a node, insert a new tree
  there, and reinsert the original subtree below the selected node at
  the distinguished and specially marked \emph{foot node} of the
  inserted tree.
  We use the \emph{non-strict} variant of TAG, in
  which the root and foot labels of the inserted tree
  need not coincide with the label of the
  replaced node to perform an adjunction.
  To control at which nodes adjunction is allowed, each node is equipped
  with two types of constraints.  The \emph{selective adjunction}
  constraint specifies a set of trees that can be adjoined and the
  Boolean \emph{obligatory adjunction} constraint specifies whether
  adjunction is mandatory.  Only trees without obligatory adjunction
  constraints are part of the generated tree language.

Figure~\ref{fig:tagExample} shows the elementary trees of an
  example TAG. Only tree~$1$ is initial and foot nodes are marked by a
  superscript asterisk~$\cdot^*$ on the label.  Whenever adjunction is
  forbidden (i.e., empty set as selective adjunction constraint and
  non"-obligatory adjunction), we omit the constraints
  altogether.  Otherwise, the constraints are put next to the
  label.  For example, $\{2,3\}^{\scriptscriptstyle+}$ indicates
  that tree $2$~or~$3$ must ($+$ = obligatory) be adjoined.

We briefly sketch the transformation from TAG to sCFTG
  that was presented by Kepser and Rogers~\cite{kepser2011equivalence}.
  TAG is a notational variant of
  footed simple CFTG, in which all variables in right-hand sides of
  productions appear in order directly below a designated \emph{foot
    node}.  To obtain an sCFTG, the footed simple CFTG is first
  converted into a spine grammar, where the spine is the path from the
  root to the foot node, and then brought into normal form using the
  construction of Fujiyoshi and Kasai~\cite{fujiyoshi2000spinal}.  The
  spine grammar of Example~\ref{ex:spinegrammar} is strongly
  equivalent to the TAG shown in Figure~\ref{fig:tagExample}.

\begin{figure}
\hspace{3em}
  \begin{tikzpicture}[emph/.style={edge from parent/.style={black,very thick,draw}},norm/.style={edge from parent/.style={black,thin,draw}}]
    \node at (-1,0) {\underline{$1$:}};
    \node at (1.45,-.7) {\tiny $\{2,3\}^{\scriptscriptstyle+}$};
    \node at (0,0) {$\alpha_2$}
      [grow'=down,level distance=2em,sibling distance=4em]
      child[norm] {node {$\gamma_2$}
      [grow'=down,level distance=2em,sibling distance=3em]
      child[norm] {node {$\gamma$}}
      child[norm] {node {$\delta$}}
    }
    child[norm] {node {$\alpha$}};
  \end{tikzpicture}
  \hfill
  \begin{tikzpicture}[norm/.style={edge from parent/.style={black,very thick,draw}},norm/.style={edge from parent/.style={black,thin,draw}}]
    \node at (-1,0) {\underline{$2$:}};
    \node at (1.45,-.7) {\tiny $\{2,3\}^{\scriptscriptstyle+}$};
    \node at (0,0) {$\alpha_2$}
      [grow'=down,level distance=2em,sibling distance=4em]
      child[norm] {node {$\gamma_2$}
      [grow'=down,level distance=2em,sibling distance=3em]
      child[norm] {node {$\gamma$}
      }
      child[norm] {node {$\gamma_2^*$}}
    }
    child[norm] {node {$\alpha$}};
  \end{tikzpicture}
\hspace{3em}

  \bigskip
\hspace{3em}
  \begin{tikzpicture}[norm/.style={edge from parent/.style={black,very thick,draw}},norm/.style={edge from parent/.style={black,thin,draw}}]
    \node at (-1,0) {\underline{$3$:}};
    \node at (.95,-.7) {\tiny $\{4\}$};
    \node at (0,0) {$\beta_2$}
      [grow'=down,level distance=2em,sibling distance=3em]
      child[norm] {node {$\beta$}}
      child[norm] {node {$\gamma_2^*$}};
  \end{tikzpicture}
  \hfill
  \begin{tikzpicture}[norm/.style={edge from parent/.style={black,very thick,draw}},norm/.style={edge from parent/.style={black,thin,draw}}]
    \node at (-1,0) {\underline{$4$:}};
    \node at (.45,0) {\tiny $\{5\}$};
    \node at (0,0) {$\eta_2$}
      [grow'=down,level distance=2em,sibling distance=3em]
      child[norm] {node {$\beta^*$}}
      child[norm] {node {$\beta$}};
  \end{tikzpicture}
  \hfill
  \begin{tikzpicture}[norm/.style={edge from parent/.style={black,very thick,draw}},norm/.style={edge from parent/.style={black,thin,draw}}]
    \node at (-1,0) {\underline{$5$:}};
    \node at (.95,-.7) {\tiny $\{5\}$};
    \node at (0,0) {$\eta_2$}
      [grow'=down,level distance=2em,sibling distance=3em]
      child[norm] {node {$\eta_2^*$}}
      child[norm] {node {$\beta$}};
  \end{tikzpicture}
\hspace{3em}
\caption{Tree-adjoining grammar}
\label{fig:tagExample}
\end{figure}
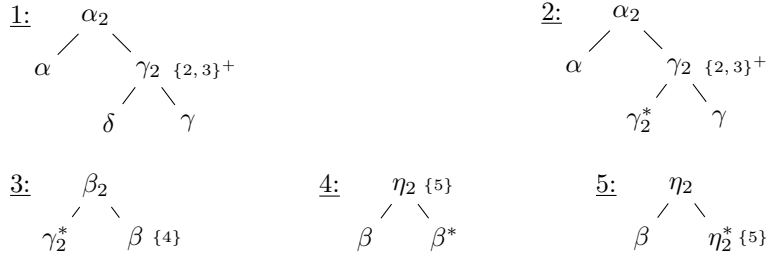

\section{Decomposition into Spines}
We proceed with the construction starting from
the normalized spine grammar~$\mathcal G$.
First, we will construct a CFG that captures all
information of~$\mathcal G$.  It represents the spinal trees (from
bottom to top) as strings and enriches the symbols with the
spine generator (initialized by start productions and
preserved by chain productions) and the non-spinal child (given by
terminal productions).  The order of these annotations depends on the
spine direction of the symbol.  The leftmost symbol of the generated
strings has only a
spine generator annotated since the bottom of the spine
has no children.  To simplify the notation, we write $n_g$ for $(n, g)
\in N^2$, $\alpha_n$ for $(\alpha, n) \in \Sigma_0 \times N$,
and $\tfrac \sigma{n_1 \;\, n_2}$ for $(\sigma, n_1, n_2) \in \Sigma_2
\times N^2$.

\begin{definition}
  \label{df:spines}
  Let spine grammar $\mathcal G$ be normalized and $\top \notin N$.
  The spines $\SPINES(\mathcal G) = L(\mathcal G')$ of~$\mathcal G$
  are the strings generated by the 
  CFG $\mathcal G' = (\{\top\} \cup N^2, \Sigma', \top, P')$ with
  $\Sigma' = (\Sigma_0 \times N) \cup (\Sigma_2 \times N^2)$ and
  productions $P' = P_0 \cup P_1 \cup P_2$ given by
  \begin{align*}
    P_0
    &= \bigl\{ \top \to \alpha_n \;\big|\; (n \to \alpha) \in P
      \bigr\} \cup
    \bigl\{ \top \to \alpha_n \: b_n
      \;\big|\; \bigl(n \to b(\alpha) \bigr) \in P \bigr\} \\[1.5ex]
    P_1
    &= \bigcup_{g \in N} \bigl\{ n_g \to b'_g \: b_g
      \;\big|\; \bigl(n \to b \bigl(b'(\SBox) \bigr) \bigr) \in P
      \bigr\} \\
    P_2
    &= \bigcup_{g \in N} \Bigl( \bigl\{n_g \to \frac \sigma {g
      \;\, n'} \;\big|\; \bigl(n \to \sigma(\SBox, n') \bigr) \in P
      \bigr\}\ \cup \\
    &\phantom{= \bigcup_{g \in N} \Bigl(}  \bigl\{ n_g \to \frac \sigma {n'\;\, g} \;\big|\; \bigl(n \to
      \sigma(n', \SBox) \bigr) \in P \bigr\} \Bigr) .
  \end{align*}
\end{definition}

\begin{example}
  \label{ex:cfg}
  We list some corresponding productions of the spine
  grammar~$\mathcal G$ (left) of Example~\ref{ex:spinegrammar} and the
  CFG~$\mathcal G'$ (right) for its spines~$\mathcal S(\mathcal G)$. 
  \begin{align*}
    \overline{a}
    &\to \alpha
    &:&
    & \top
    &\to \alpha_{\overline a} \\*
    s
    &\to t(\delta)
    &:&
    & \top
    &\to \delta_s t_s \\*
    t
    &\to a \bigl(b'(\SBox) \bigr)
    &:&
    & t_s
    &\to b_s' a_s
    & t_{a}
    &\to b_{a}' \, a_{a}
    & t_{\overline{b}}
    &\to b_{\overline{b}}' \, a_{\overline{b}}
    & \dots \\*
    a
    &\to \alpha_2(\overline{a}, \SBox)
    &:&
    & a_s
    &\to \frac{\alpha_2}{\overline{a} \quad s}
    & a_{a}
    &\to \frac{\alpha_2}{\overline{a} \quad a}
    & a_{\overline b}
    &\to \frac{\alpha_2}{\overline{a} \quad \overline{b}}
    & \dots
  \end{align*}

  Note that for each start production we obtain a single production
  since the nonterminal on the left"-hand side becomes the spine generator.
  On the other hand, for each chain or terminal production we have to
  combine them with all nonterminals, as we do not know the spine
  generator of the nonterminal on the left"-hand side of the original
  production.  When a string is derived, the spine generators are pulled
  through originating from start productions and are consistent
  throughout the string.
  The language generated by $\mathcal G'$ is
  \begin{align*}
  \mathcal \SPINES(\mathcal G) = \Big\{
  \delta_s\  \ANNOC^n \, \ANNOB \  \ANNOA^n \,\Big|\,
  n \in \natp \Big\}\, \cup\,
  \Big\{
  \beta_{\overline{b}} \  \ANNOE^m \,\Big|\,
  m \in \nat \Big\}\, \cup \, \Big\{
  \alpha_{\overline{a}},
  \beta_{\overline{e}},
  \gamma_{\overline{c}}
      \Big\} \, .\\[-1.4em]
  \end{align*}

\end{example}

Note that each string generated by the CFG belongs to~$(\Sigma_0
\times N) (\Sigma_2 \times N^2)^*$.  Next we define how to
reassemble those spines to form trees again, which then relabel to the
original trees generated by~$\mathcal G$.  The operation given in the
following definition describes how a string generated by the CFG can
be transformed into a tree by attaching subtrees in the non-spinal
direction of each symbol, whereby the non-spinal child annotation of the
symbol and the spinal annotation of the root of the attached tree have to match.

\begin{definition}
  \label{df:CombineSpines}
  Let $T \subseteq \TREES{\Sigma_2 \times N^2, \emptyset}(\Sigma_0
  \times N)$ and let $w \in A = (\Sigma_0 \times N) (\Sigma_2 \times N^2)^*$.
  The generator
  $\GENERATOR \colon (\Sigma_0 \times N) \cup (\Sigma_2 \times N^2
  ) \to N$ is the nonterminal in spine direction and is given by
  \[ \GENERATOR (a) =
    \begin{cases}
      n
      & \text{if } a = \alpha_n \in \Sigma_0 \times N \\*
      n_{d(\sigma)}
      & \text{if } a = \frac \sigma {n_1 \;\, n_2} \in \Sigma_2 \times
      N^2 \enspace.
    \end{cases} \]
  For $n \in N$, let $T_n = \bigl\{t \in T \;\big|\;
  \GENERATOR \bigl(t(\varepsilon) \bigr) = n \bigr\}$ 
  be those trees of~$T$ whose root label has~$n$ annotated in spine
  direction.  We recursively define the tree language $\attach_T(w) \subseteq
  \TREES{\Sigma_2 \times N^2, \emptyset}(\Sigma_0 \times N)$
  by $\attach_T(\alpha_n) = \{\alpha_n\}$ for all $\alpha_n \in
  \Sigma_0 \times N$, and
  \begin{align*}
    \attach_T \bigl(w \, \frac \sigma {n_1\;\, n_2}
      \bigr)  
    = \Biggl\{\frac \sigma {n_1\;\, n_2}(t_1, t_2) \;\Bigg|\;
      \begin{matrix}
         \,t_{d(\sigma)}   & \in & \attach_T(w) \\
         \,t_{3-d(\sigma)} & \in & T_{n_{3-d(\sigma)}}
      \end{matrix}
    \Biggr\}
  \end{align*}
  for all $w \in A$ and $\tfrac \sigma {n_1\;\, n_2} \in \Sigma_2
  \times N^2$.
\end{definition}

To obtain the tree language defined by~$\mathcal G$, it is necessary
to apply this operation recursively on the set of spines.

\begin{definition}
  \label{df:CombineSpines2}
  Let $L \subseteq (\Sigma_0 \times N) (\Sigma_2 \times N^2
  )^*$.  We inductively define the tree 
  language~$\FORESTS(L)$ generated by~$L$ as the smallest tree
  language~$\FORESTS$ with~$\attach_{\FORESTS}(w) \subseteq
  {\FORESTS}$ for every~$w \in L$.
\end{definition}

\begin{example}
\label{ex:reassembly}
The CFG $\mathcal G'$ of Example~\ref{ex:cfg} generates the set of spines
$\SPINES(\mathcal G)$ and
$\mathcal F(\SPINES(\mathcal G))_S$ contains the correctly
assembled trees formed from these spines.
Figure~\ref{fig:assembledTree} shows a tree of $\mathcal
F(\SPINES(\mathcal G))_S$ since the generator of the main spine is $S
= s$, which is stored in spinal direction in the root label $\ANNOA$.
We can observe the correspondence of annotations in non-spinal direction
and the spine generator of the respective child in the
same direction.
\end{example}

Next we prove that ${\FORESTS} \bigl(\SPINES(\mathcal G)
\bigr)_S$~and~$T(\mathcal G)$ coincide modulo relabeling.  This shows
that the context-free language~$\SPINES(\mathcal G)$ of spines
completely describes the tree language~$T(\mathcal G)$ generated
by~$\mathcal G$.

\begin{theorem}
  \label{thm:spines}
  Let spine grammar~$\mathcal G$ be normalized.  Then $\pi \bigl({\FORESTS}
  (\SPINES(\mathcal G))_S \bigr) = T(\mathcal G)$, where the
  relabeling $\pi \colon (\Sigma_0 \times N) \cup (\Sigma_2 \times N^2
  ) \to \Sigma_0 \cup \Sigma_2$ is given by~$\pi(\alpha_n) =
  \alpha$ and $\pi\bigl(\tfrac \sigma {n_1 \;\, n_2}\bigr) = \sigma$ for
  all~$\alpha \in \Sigma_0$, $\sigma \in \Sigma_2$, and $n, n_1, n_2
  \in N$.
\end{theorem}

\begin{proof}
    We will prove a more general statement.  Let $\mathcal G'$ be the
    CFG constructed for~$\mathcal G$ in
    Definition~\ref{df:spines}.  Given~$n \in N_0$, we
    will show that the tree language~$\pi({\FORESTS}(\SPINES(\mathcal
    G))_n)$ coincides with~$\{t \in \TREES{\Sigma_2,
      \emptyset}(\Sigma_0) \mid n \DERIVES_{\mathcal{G}}^* t 
    \}$, which contains the trees that can be derived in~$\mathcal G$
    starting from the nullary nonterminal~$n$.  To this end, we
    show inclusion in both directions.

    The inclusion~$(\subseteq)$ is proved by
    induction on the size of $t \in {\FORESTS}(\SPINES(\mathcal
    G))_n$.  Clearly, $t$~was constructed from a string $w =
    \alpha_n\,
    \tfrac {\sigma_1} {n_{1,1} \;\, n_{1,2}}\,
    \cdots\,
    \tfrac {\sigma_m} {n_{m,1} \;\, n_{m,2}}
    \in \SPINES(\mathcal G)$ with spine generator~$n$
    and $n_{i, d(\sigma_i)} = n$ for all~$i \in \{1, \dotsc, m \}$. 
    Hence, there is a derivation~$\top \DERIVES_{\mathcal G'}^* w$,
    where $\top$ is the start nonterminal of~$\mathcal G'$.
    Each production applied during this derivation corresponds
    uniquely to a production of the spine grammar~$\mathcal G$.
    This yields a derivation~$n \DERIVES_{\mathcal G}^* t_w$ of a
    spinal tree~$t_w \in I_{\mathcal G}(n)$, where the spine of~$t_w$
    is labeled (from bottom to top) by~$\pi(w)$.  Besides the spine,
    $t_w$~contains only nullary nonterminals that 
    for~$i \in \{1, \dotsc, m\}$
    are attached below~$\sigma_i$  in the non-spinal
    direction~$3-d(\sigma_i)$ and are labeled
    by~$\pi_1(n_{i,{3-d(\sigma_i)}})$, respectively.
    For better readability, let
    $n_i = n_{i,{3-d(\sigma_i)}}$ in the following.
    Each nonterminal annotation~$n_i$ in~$w$ implies the attachment
    of a tree~$t_i \in {\FORESTS}(\SPINES(\mathcal G))_{n_i}$.  These
    attached trees are smaller than~$t$, so we can use the induction
    hypothesis and conclude that there is a derivation~$n_i
    \DERIVES_{\mathcal G}^* \pi(t_i)$.  Combining those derivations, 
    we obtain a derivation~$n \DERIVES_{\mathcal G}^* \pi(t)$.
    
    To prove the other direction~$(\supseteq)$, we use induction on the
    length of the derivation~$n \DERIVES_{\mathcal G}^* t$ and show
    that there exists a tree~$t' \in {\FORESTS}(\SPINES(\mathcal
    G))_n$ with $\pi(t') = t$.  To this end, we reorder the
    derivation such that a spinal tree~$s \in I_{\mathcal G}(n)$ is
    derived first (i.e., $n \Rightarrow_{\mathcal G}^* s
    \Rightarrow_{\mathcal G}^* t$).
    Suppose that this spinal tree~$s$ has the nullary terminal symbol~$\alpha$ at the
    bottom and contains $m$~binary terminal symbols~$\seq \sigma1m$ (from bottom
    to top).
    Let $n_i$~be the non-spinal child of~$\sigma_i$.  It is 
    attached in direction~$3 - d(\sigma_i)$.  Due to the construction
    of~$\mathcal G'$, there is a corresponding derivation $\top
    \DERIVES_{\mathcal G'}^* w$, for which the derived string~$w \in
    \SPINES(\mathcal G)$ has the form $w =
    \alpha_n\,
    \tfrac {\sigma_1} {n_{1,1} \;\, n_{1,2}}\,
    \cdots\,
    \tfrac {\sigma_m} {n_{m,1} \;\, n_{m,2}}$
    with $n_{i,d(\sigma_i)} = n$
    and $n_{i,3-d(\sigma_i)} = n_i$ for all~$i \in \{ 1, \dotsc, m
    \}$.  The remaining nonterminals in~$s$ are replaced by
    subderivations $n_i \DERIVES_{\mathcal G}^* t_i$ for all~$i \in
    \{1, \dotsc, m\}$.  These subderivations are shorter than the overall
    derivation~$n \DERIVES_{\mathcal G}^* t$, so by the induction
    hypothesis, there exist trees $t'_i \in {\FORESTS}(\SPINES(\mathcal
    G))_{n_i}$ such that $\pi(t'_i) = t_i$ for all~$i \in \{1,
    \dotsc, m\}$.  Attaching those trees~$t'_i$ to~$w$, we 
    obtain the tree~$t' \in {\FORESTS}(\SPINES(\mathcal G))_n$.  As
    required, we have $\pi(t') = t$.

    We thus have proved that $\pi({\FORESTS}(\SPINES(\mathcal
    G))_S)$ coincides with $\{t \in \TREES{\Sigma_2,
      \emptyset}(\Sigma_0) \mid S \DERIVES_{\mathcal{G}}^* t\}$.  
    Hence, $\pi({\FORESTS}(\SPINES(\mathcal G))_S) = T(\mathcal
    G)$.
\end{proof}

\begin{corollary}
  \label{cor:spines}
  Let $\SHORTWORDS = \bigl\{w \in \Next
  \bigl(\SPINES(\mathcal G) \bigr) \mid \abs{w} = 1 \bigr\}$.
  Then there exists a pop"-normalized MPDA~$\mathcal A$ with
  $L(\mathcal A) \cup \SHORTWORDS  = \Next \bigl(\SPINES(\mathcal G)
  \bigr)$.
  Moreover, the tree languages ${\FORESTS} \bigl( L(\mathcal A) \cup  L_1
  \bigr)_S$~and~$T(\mathcal G)$ coincide modulo relabeling.
\end{corollary}

\begin{proof}
  Without loss of generality (see Theorem~\ref{thm:normalform}), let
  $\mathcal G = (N, \Sigma, S, P)$ be a normalized spine
  grammar.  Clearly, $\SPINES(\mathcal G)$ is a context"-free subset
  of~$A_0A_2^*$ with $A_0 = \Sigma_0 \times N$ and $A_2 = \Sigma_2
  \times N^2$ by Definition~\ref{df:spines}.
  Corollary~\ref{cor:MPDA} yields a pop"-normalized MPDA~$\mathcal A$
  such that~$L(\mathcal A) = \bigl\{ w \in \Next
  \bigl(\SPINES(\mathcal G) \bigr) \mid  \abs{w} \geq 2 \bigr\}$.
  Moreover, we observe that $L(\mathcal A) \cup L_1 \subseteq (A'_2
  \times A_0)(A'_2 \times A_2)^*$ with $A'_2 = A_2 \cup
  \{\triangleleft\}$.  Clearly, $L(\mathcal A) \cup L_1$ relabels
  to~$\SPINES(\mathcal G)$ via the projection to the components of
  $A_0$~and~$A_2$.  Consider the ranked alphabet~$\Sigma'$ given 
  by~$\Sigma'_0 = A'_2 \times \Sigma_0$, $\Sigma'_1 = \emptyset$, and
  $\Sigma'_2 = A'_2 \times \Sigma_2$.  Then ${\FORESTS}
  \bigl( L(\mathcal A) \cup L_1 \bigr) \subseteq \TREES{\Sigma'_2
    \times N^2, \emptyset}(\Sigma'_0 \times N)$
  relabels to~${\FORESTS} \bigl(\SPINES(\mathcal G) \bigr) \subseteq
  \TREES{\Sigma_2 \times N^2, \emptyset}(\Sigma_0 \times N)$ via the
  projection to the components of $\Sigma_0 \times N$~and~$\Sigma_2
  \times N^2$.  By Theorem~\ref{thm:spines}, the tree language~${\FORESTS}
  \bigl(\SPINES(\mathcal G) \bigr)_S$ relabels to~$T(\mathcal G)$,
  which proves that ${\FORESTS} \bigl(L(\mathcal A) \cup L_1
  \bigr)_S$~and~$T(\mathcal G)$ coincide modulo relabeling.
\end{proof}

\begin{figure}
  \centering
  \begin{tikzpicture}[shorten >=1pt,node distance=1.7cm,on grid,auto,
    scale=0.8, every node/.style={scale=0.8}] 
    \node[state,initial] (q_0)   {$q_0$}; 
    \node[state] (q_1) [right=of q_0] {$q_1$}; 
    \node[state] (q_1') [below=of q_1] {$q_1'$}; 
    \node[state] (q_2) [right=of q_1] {$q_2$};
    \node[state] (q_3) [right=of q_2] {$q_3$}; 
    \node[state,accepting](q_3') [below=of q_3] {$q_3'$};
    \path[->] 
      (q_0) edge [] node {} (q_1)
      (q_0) edge [] node {} (q_1')
      (q_1) edge [loop above] node {$\upsilon \downarrow$} ()
      (q_1) edge [] node [] {$\upsilon \downarrow$} (q_1')
      (q_1') edge [] node {} (q_2)
      (q_2) edge [] node {$\upsilon \uparrow$} (q_3)
      (q_3) edge [loop above] node {$\upsilon \uparrow$} ()
      (q_3) edge [] node {$\omega \uparrow$} (q_3')
      (q_2) edge [] node [below left] {$\omega \uparrow$} (q_3');

    \node (anchor) [right=of q_3] {};

    \node[state,initial] (p_0) [right=of anchor]   {$p_0$}; 
    \node[state] (p_1) [right=of p_0] {$p_1$}; 
    \node[state,accepting] (p_1') [below=of p_1] {$p_1'$}; 
    \path[->] 
      (p_0) edge [] node {} (p_1)
      (p_0) edge [below left] node {$\chi \uparrow$} (p_1')
      (p_1) edge [loop right] node {} ()
      (p_1) edge [] node [] {$\chi \uparrow$} (p_1');
  \end{tikzpicture}
  \caption{Sample MPDA (see Example~\protect{\ref{ex:mpda}})}
  \label{fig:mpda}
\end{figure}
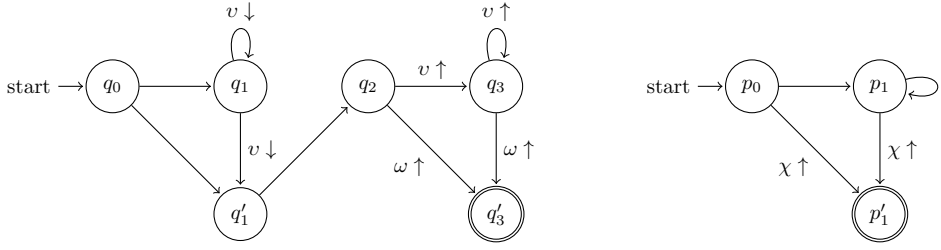

\begin{example}
  \label{ex:mpda}
  The MPDA constructed in Corollary~\ref{cor:spines} for the spine
  grammar~$\mathcal G$ of Example~\ref{ex:spinegrammar} is depicted in
  Figure~\ref{fig:mpda}.
  Initial states are indicated using a start marker
  and final states are marked by a double circle.  Push and pop stack
  operations are written with downwards and upwards arrows, respectively.
  The MPDA consists of two components.  The larger
  one describes the main spine, and the smaller one describes the side
  spine.  The distinction between the three stack symbols is necessary 
  due to pop-normalization, and the distinction between $q_1$~and~$q'_1$
  (and similar states) is necessary because
  of the lookahead provided by $\Next\bigl(\SPINES(\mathcal G) \bigr)$.
  For example, $\tau(q_1) =
  \bigl(\tfrac{\gamma_2}{s \quad \overline{c}}, \tfrac{\gamma_2}{s
    \quad \overline{c}} \bigr)$ and $\tau(q'_1) =
  \bigl(\tfrac{\beta_2}{s \quad \overline{b}}, \tfrac{\gamma_2}{s
    \quad \overline{c}} \bigr)$.  Similarly, $\tau(p_1) = (z, z)$ and
  $\tau(p'_1) = (\triangleleft, z)$ where $z =
  \tfrac{\eta_2}{\overline{e} \quad \overline{b}}$.  To completely
  capture the behavior of~$\mathcal G$, we additionally require the
  set $L_1 = \{(\triangleleft, \alpha_{\overline{a}}),
  (\triangleleft, \beta_{\overline{b}}), (\triangleleft,
  \beta_{\overline{e}}), (\triangleleft, \gamma_{\overline{c}})\}$,
  which contains the spines of length~$1$.
\end{example}

\section{CCG Construction}
\label{sec:constructingccg}
In this section, let $\mathcal G = (N, \Sigma, S, P)$ be a normalized
spine grammar with spine direction $d \colon \Sigma \to \{1, 2\}$ and
$\mathcal A = (Q, \Delta, \Gamma, \delta, \tau, I, F)$ the
pop-normalized MPDA constructed in Corollary~\ref{cor:spines} with
$\mathord{\RETURN} \colon \Gamma \to Q$.  We note that $\Delta =
\Sigma' \times \Sigma''$ with $\Sigma' =  \{\triangleleft\} \cup
(\Sigma_2 \times N^2)$ as well as $\Sigma'' = (\Sigma_0 \times N) \cup
(\Sigma_2 \times N^2)$.  Moreover, let $\bot \notin Q$ be a special
symbol.  To provide better access to the components of the
MPDA~$\mathcal A$, we define some additional maps.

The spine generator $\GENERATOR \colon Q \to N$ is given
by $\GENERATOR(q) = \GENERATOR(s_2)$ for every state $q \in Q$,
where $\tau(q) = (s_1, s_2) \in \Delta$.
Since $\AUTOMATON$~cannot accept
strings of length~$1$, we have to treat them separately.  Let
$\SHORTWORDS = \bigl\{w \in \Next \bigl(\SPINES(\mathcal G) \bigr)
\mid \abs{w} = 1 \bigr\}$ and $\GENERATOR \colon L_1 \to N$ be given
by $\GENERATOR(w) = n$ for all $w = (\triangleleft, \alpha_n) \in
L_1$.
We extend $\tau \colon Q \to
\Delta$ to $\tau' \colon (Q\, \cup \,\SHORTWORDS) \to \Delta$ by $\tau'(q)
= \tau(q)$ for all $q \in Q$ and $\tau'(a) = a$ for short strings~$a
\in L_1$.

Recall that $D = \{\mathord{\SLASHF}, \mathord{\SLASHB}\}$.
Given state $q \in Q \setminus F$, the slash
type $\SLASHDIR \colon (Q \setminus F) \to D$
states whether the symbol $\tau(q)$ produced by state~$q$
occurs as the first or second child of its parent symbol.
The combining nonterminal $\COMB \colon (Q \setminus F) \cup \{\bot\} \to N$
denotes with which spine generator the symbol $\tau(q)$ is combined.
Let $\tau(q) = \bigl(\tfrac \sigma {n_1 \;\, n_2}, s_2
\bigr)$ with $\tfrac \sigma {n_1 \;\, n_2} \in \Sigma_2 \times N^2$
and $s_2 \in \Sigma''$.  The slash type and the combining nonterminal
can be determined from the next symbol~$\tfrac \sigma {n_1 \;\,
  n_2}$.  Formally, $\SLASHDIR(q) = \mathord{\SLASHF}$ if $d(\sigma) =
1$ and $\SLASHDIR(q) = \mathord{\SLASHB}$ otherwise.  Further,
$\COMB(q) = n_{3-d(\sigma)}$ and $\COMB(\bot) = S$.

We simulate the accepting runs of~$\AUTOMATON$ in the spines
consisting of primary categories of the CCG.
These \emph{primary spines} are paths in a CCG derivation tree that
start with a lexical category at a leaf of the derivation tree and
consist of a (possibly empty) sequence of primary categories followed
by a secondary or initial category.
The shortest possible primary spine is a single lexical category that serves
as a secondary or initial category.
The main idea is
that the primary categories on the primary spine store the current
configuration of~$\AUTOMATON$.  This is achieved by adding an
additional argument for transitions that push a symbol, whereas for
each pop transition, an argument is removed.  The last
argument stores the current state in the first component and the top
of the stack in the second component.  The previous arguments store
the preceding stack symbols in their second components and the state 
the automaton returns to when the stack symbol stored in the next
argument is popped in the first components.  To implement the required
transformations of consecutive primary categories, the secondary
categories need to have a specific structure.  This mandates
that the categories at the top of a spine (which act as secondary
categories unless they belong to the main spine) cannot store their corresponding automaton state in
the first component of the last argument as usual, but instead utilize
the third component of their target.  Thus each argument stores the
final state corresponding to its secondary combination partner
in the third component.
This third component also allows us to
decide whether a category is primary:
A category is a primary category if and only if the
spine generator of the state stored in the first component of the last
argument and the spine generator of the state stored in the last
component of the target coincide.
This is possible since $\mathcal G$ is normalized, which yields 
that attaching spines have a spine generator that is different from the
spine generator of the spine that they attach to.

\begin{definition}
  \label{df:simccg}
  We define the \CCG\ $\mathcal G_{\mathcal A, L_1} = (\Delta_0, A,
  R, I', {\LEXICON})$ as follows:

  Let $A = \{(q, \gamma, f) \in A' \mid \GENERATOR(f) = \COMB(q)\}$
  with $A' = (Q \cup \{\bot\}) \times \Gamma^{\leq 1} \times (F \cup
  \SHORTWORDS)$.  We use~$a_i$ to refer to the $i$"~th component
  of an atom~$a \in A$.  Additionally, let $I' = \{ (\bot,
  \varepsilon, f) \in A \mid \GENERATOR(f) = S\}$.
    
  In the rules~$R = \bigcup_{\SLASH \in D} (R_1^{\SLASH} \cup
  R_2^{\SLASH} \cup R_3^{\SLASH})$ we underline the primary category
  $ax \SLASHF b$, which always needs to fulfill $\GENERATOR(a_3) =
  \GENERATOR(b_1)$.
  \begin{align*}
    R_1^{\SLASHF}
    &= \hspace{-.5em} \bigcup_{\substack{a, b, c \in A,\,
      \mathord{\SLASH} \in D \\ (b_1, \varepsilon, \varepsilon, c_1)
    \in \delta \\ b_2 = c_2}} \Biggl\{ \frac{ax \,\SLASH\,
    c}{\underline{ax \SLASHF b} \qquad b \,\SLASH\, c} \Biggr\}
    \qquad \quad
    R_2^{\SLASHF}
    = \hspace{-1em} \bigcup_{\substack{a, b, c, e \in A,\,
      \mathord{\SLASH}, \mathord{\SLASH'} \in D \\ (b_1,
    \varepsilon, e_2, e_1) \in \delta \\ b_2 = c_2 \\ c_1 =
    \RETURN(e_2)}} \Biggl\{ \frac{ax \,\SLASH\, c \,\SLASH'\,
    e}{\underline{ax \SLASHF b} \qquad b \,\SLASH\, c \,\SLASH'\, e}
    \Biggr\} \\*
    R_3^{\SLASHF}
    &= \hspace{-.5em} \bigcup_{\substack{a, b \in A \\ (b_1, b_2,
    \varepsilon, q) \in \delta}} \Biggl\{ \frac{ax}{\underline{ax
    \SLASHF b} \qquad b} \Biggr\}
  \end{align*}
  We listed all the forward rules, but for each forward rule there
  also exists a symmetric backward rule yielding the rule
  sets~$R_1^{\SLASHBS}$, $R_2^{\SLASHBS}$, and~$R_3^{\SLASHBS}$.

  We require some notions for the lexicon.
  A category $c \in \CATEGORIES{A}$ is \emph{well-formed} if we have
  $\mathord{\SLASH} = \SLASHDIR(b_1)$, $b_1 \in Q$, and $b_2 \in \Gamma$
  for every $i \in [\ARITY(c)]$ with $\SLASH\, b = \ARGUMENT(c, i)$.  Let $\RESTRCAT =
  \{ c \in \CATEGORIES{A} \mid c \text{ well-formed}\, \}$ be the set of
  well-formed categories.  Clearly $I' \subseteq \RESTRCAT$.  In
  addition, we introduce sets $\TOPCATS$ and $\TOPCATL$ of
  top-of-spine categories derived from the short strings of~$L_1$ and
  the strings accepted by~$\mathcal A$, respectively: 
  \begin{align*}
    \TOPCATS
    &= \{ a \in I' \mid a_3 \in \SHORTWORDS \}
    \, \cup \bigcup_{\substack{r \in R \\ ax =
    \SECONDARY(r)}} \{ ax \in \RESTRCAT \mid a_3 \in \SHORTWORDS \} \\*
    \TOPCATL
    &= \{ a \in I' \mid a_3 \in F\} \phantom{_1}
    \, \cup \bigcup_{\substack{r \in R \\ ax =
    \SECONDARY(r)}} \{ ax \in \RESTRCAT \mid a_3 \in F\}
  \end{align*}
  Note that $\TOPCATS \cup \TOPCATL \subseteq \RESTRCAT$.
  For all $\alpha \in \Delta_0 = \Sigma' \times ( \Sigma_0 \times N )$
  we define the lexicon as follows :
  \begin{align*}
    {\LEXICON}(\alpha)
    &=
    \left\{ ax \;\middle\vert\;
      \begin{matrix}
    ax \in \TOPCATS\\
    \tau'(a_3) = \alpha
    \end{matrix}
    \right\} 
    \, \cup\,
    \left\{ax \,\SLASH\, b \in \RESTRCAT \;\middle\vert\;
      \begin{matrix}
        ax \in \TOPCATL,&
        \GENERATOR(a_3) = \GENERATOR(b_1)\\
        \tau'(b_1) = \alpha,&
        \RETURN(b_2) = a_3 \\
        b_1 \in I
      \end{matrix}
    \right\} \\[-1em]
  \end{align*}
\end{definition}

Each atom of~$A$ consists of three components.  The first component
stores the current state of~$\AUTOMATON$ (or the special
symbol~$\bot$), the second component stores the current symbol at the
top of the stack, and the third component stores the final state
or the symbol
corresponding to the combining category of the attaching side 
spine.  With this intuition, the rule system directly implements the
transitions of~$\AUTOMATON$.

The lexicon assigns categories to symbols that can label leaves,
so these symbols are taken from the nullary terminal symbols.
The assigned categories consist of a category
that appears at the top of
a spine and an additional argument for the initial state of
an accepting run.  The spines of length~$1$ are translated directly to
secondary categories or initial categories.

Let us make two general observations that hold for all
categories that appear in derivation trees
of~$\mathcal{G}_{\AUTOMATON,L_1}$:
\begin{enumerate}
\item
All categories are well"-formed.
This follows from the fact that only well"-formed categories occur in the
lexicon and all categories in the derivation trees consist of atoms
and arguments that were already present in the
lexicon~\cite[Lemma~3.1]{vijayshanker1994equivalence}.
\item
All primary categories $ax \,\SLASH\, b$ obey $\GENERATOR(a_3) =
\GENERATOR(b_1)$.  This is directly required by the rule system. 
\end{enumerate}

Finally, we will now describe how to relabel the derivation 
trees~$\mathcal{D}(\mathcal{G}_{\AUTOMATON,L_1})$ of the
CCG~$\mathcal{G}_{\AUTOMATON,L_1}$ that uses categories built using the
input symbols of the MPDA~$\mathcal{A}$.  Note that only well"-formed
categories will occur in derivation trees.  Primary and non"-primary
categories are relabeled differently.    The relabeling $\rho \colon \RESTRCAT \to \Delta$ is
defined for every $c \in \RESTRCAT$ by $\rho(ax \,\SLASH\, b) =
\STATEOUTPUTEXTEND(b_1)$ for all primary categories $ax \,\SLASH\, b
\in \RESTRCAT$; i.e., $\GENERATOR(a_3) = \GENERATOR(b_1)$.  Otherwise
$\rho(ax) = \STATEOUTPUTEXTEND(a_3)$ for all initial and secondary categories~$ax
\in \RESTRCAT$.

The following property requires that the spine
grammar~$\mathcal G$ is normalized, so a spine never has the same
spine generator as its attached spines.

\begin{lemma}
   For all secondary categories
  $ax \,\SLASH\, b$ we have $\GENERATOR(a_3) \neq \GENERATOR(b_1)$.
\end{lemma}

\begin{proof}
  We have~$\GENERATOR(a_3) = \COMB(a_1)$ by the definition of atoms~$A$.
  Additionally, we have $\GENERATOR(a_1)
  = \GENERATOR(b_1)$ by the construction of the rule system,
  since $a_1, b_1 \in Q$ occur in single transition of~$\mathcal A$.
  However, the spine generator~$\GENERATOR(a_1)$
  never coincides with the spine generator~$\COMB(a_1)$
  of an attaching spine due to normalization of $\mathcal G$.
  So $\COMB(a_1) \neq \GENERATOR(a_1)$ and thus $\GENERATOR(a_3) =
  \COMB(a_1) \neq \GENERATOR(a_1) = \GENERATOR(b_1)$.
\end{proof}

We are now ready to describe the general form of primary spines
of~$\mathcal{G}_{\AUTOMATON, L_1}$.  Given a primary spine~$c_0
\cdots c_n$ with~$n \geq 1$ read from lexicon entry towards the root,
we know that it starts with a lexicon entry $c_0 = ax \,\SLASH\, b \in
\LEXICON(\Delta_0)$ and ends with the non"-primary category~$ax$,
which as such cannot be further modified.
Hence each
of the categories~$c \in \{\seq c0{n-1}\}$ has the form~$ax
\,\SLASH_1\, b_1 \cdots \,\SLASH_m\, b_m$ with $m \in \natp$.  Let $b_i =
(q_i, \gamma_i, f_i)$ for every $i \in [m]$.  The category~$c_n$ is
relabeled to~$\STATEOUTPUTEXTEND(a_3)$ and $c$~is relabeled
to~$\STATEOUTPUTEXTEND(q_m)$.  Additionally, unless $a_1 = \bot$,
the first components of all
atoms in~$ax$ have the same spine generator $\GENERATOR(a_1)$ and
$\GENERATOR(q_1) = \cdots = \GENERATOR(q_m)$, but $\GENERATOR(a_1)
\neq \GENERATOR(q_1)$.  Finally, neighboring arguments~$\SLASH_{i-1}\,
b_{i-1} \,\SLASH_i\, b_i$ in the suffix are coupled such
that~$\RETURN(\gamma_i) = q_{i-1}$ for all~$i \in [m] \setminus
\{1\}$.  This coupling is introduced in the rules of second degree and
preserved by the other rules. 

Using these observations, it can be proved that the primary spines
of~$\mathcal G_{\AUTOMATON, L_1}$ are relabeled to strings
of~$\Next \bigl(\SPINES(\mathcal G) \bigr)$ and vice versa.
Additionally, spines attach in essentially the same manner in the CCG
and using~$\FORESTS$.
This yields the result that, given a spine grammar, it is possible to
construct a CCG that generates the same tree language.
We will prove the correctness of our construction in the following two subsections.
But first, we will illustrate it by means of an example.

\newcommand{\icol}[1]{
  \left(\begin{smallmatrix}#1\end{smallmatrix}\right)}

\newcommand{\icolt}[3]{
  \left(\begin{smallmatrix}#1\vphantom{q_3'}\\#2\vphantom{q_3'}\\#3\vphantom{q_3'}\end{smallmatrix}\right)}

\newcommand{\irow}[1]{
  \begin{smallmatrix}(#1)\end{smallmatrix}}

\begin{figure*}
  \resizebox{\textwidth}{!}{
  \raisebox{\depth}{\scalebox{1}[-1]{
      $\INFER{\ROT{\icolt{\bot}{\varepsilon}{q_3'} \SLASHB
          \icolt{q_3}{\omega}{\alpha} \SLASHB
          \icolt{q_2}{\upsilon}{\alpha}}}{%
        \INFER{\ROT{\icolt{\bot}{\varepsilon}{q_3'} \SLASHB
            \icolt{q_3}{\omega}{\alpha} \SLASHF
            \icolt{q_1'}{\upsilon}{p_1'}}}{
          \INFER{\ROT{\icolt{\bot}{\varepsilon}{q_3'} \SLASHF
              \icolt{q_1}{\omega}{\gamma}}}{
            \PROJECT[0]{\ROT{\icolt{\bot}{\varepsilon}{q_3'} \SLASHF
                \icolt{q_0}{\omega}{\gamma}}}{}
            &
            \PROJECT[0]{\ROT{\icolt{q_0}{\omega}{\gamma} \SLASHF
                \icolt{q_1}{\omega}{\gamma}}}{}}
          &
          \PROJECT[0]{\ROT{\icolt{q_1}{\omega}{\gamma} \SLASHB
              \icolt{q_3}{\omega}{\alpha} \SLASHF
              \icolt{q_1'}{\upsilon}{p_1'}}}{}}
        &
        \INFER{\ROT{\icolt{q_1'}{\upsilon}{p_1'} \SLASHB
            \icolt{q_2}{\upsilon}{\alpha}}}{
          \PROJECT[0]{\ROT{\icolt{p_0}{\chi}{\beta}}}{}
          &
          \PROJECT[0]{\ROT{\icolt{q_1'}{\upsilon}{p_1'} \SLASHB
              \icolt{q_2}{\upsilon}{\alpha} \SLASHB
              \icolt{p_0}{\chi}{\beta}}}{}}}$
    }}
  }
  \vspace{-3ex}
  \caption{Part of a derivation tree of~$\mathcal G_{\AUTOMATON,
      L_1}$ (see Example~\protect{\ref{ex:ccgconstr}})}
  \label{fig:ccgconstr}
\end{figure*}

\begin{example}
  \label{ex:ccgconstr}
  Figure~\ref{fig:ccgconstr} shows part of the derivation tree of
  CCG~$\mathcal G_{\AUTOMATON, L_1}$ that corresponds to the tree of
  Figure~\ref{fig:genTree}, which is generated by the spinal
  grammar~$\mathcal G$ of Example~\ref{ex:spinegrammar}.  We use the
  following abbreviations: $\alpha = (\triangleleft,
  \alpha_{\overline{a}})$, $\beta = (\triangleleft,
  \beta_{\overline{b}})$, and $\gamma = (\triangleleft,
  \gamma_{\overline{c}})$.  The labeling of the depicted section is
  $\delta\, \gamma_2\, \gamma_2\, \beta_2$ for the main spine and
  $\beta\, \eta_2$ for the side spine (see Figure~\ref{fig:genTree}).
  The corresponding runs of~$\AUTOMATON$ are $\bigl(\langle q_0,
  \omega \rangle, \langle q_1, \omega \rangle, \langle q_1', \upsilon
  \omega \rangle, \langle q_2, \upsilon \omega \rangle \bigr)$ and
  $\bigl(\langle p_0, \chi \rangle, \langle p_1', \varepsilon \rangle
  \bigr)$ (see Example~\ref{ex:mpda}, Figure~\ref{fig:mpda}).

  Let us observe how the transitions of~$\AUTOMATON$ are simulated
  by~$\mathcal G_{\AUTOMATON, L_1}$.  The first transition $(q_0,
  \varepsilon, \varepsilon, q_1)$ on the main spine does not modify
  the stack.  It is implemented by replacing the last argument
  $\SLASHF (q_0, \omega, \gamma)$ by $\SLASHF (q_1, \omega, \gamma)$.
  The next transition $(q_1, \varepsilon, \upsilon, q_1')$ pushes the
  symbol~$\upsilon$ to the stack.  The argument $\SLASHF (q_1, \omega,
  \gamma)$ is thus replaced by two arguments $\SLASHB (q_3, \omega, \alpha) \SLASHF
  (q'_1, \upsilon, p'_1)$.  As the stack grows, an additional argument
  with the new state and stack symbol is added.  The previous argument
  stores $\RETURN(\upsilon) = q_3$ to ensure that we enter the correct
  state after popping~$\upsilon$.  It also contains the previous
  unchanged stack symbol~$\omega$.  The pop transition $(p_0,
  \chi, \varepsilon, p_1')$ on the side spine run is realized by
  removing $\SLASHB (p_0, \chi, \beta)$.

  The third components are required to relabel the non-primary
  categories.  At the bottom of the main spine, $c_1 = (\bot,
  \varepsilon, q_3') \SLASHF (q_0, \omega, \gamma)$ is a primary
  category because $q_0$ and $q_3'$ are associated with the same
  spine generator $s$.  Thus, $c_1$ gets relabeled to $\tau'(q_0)$.
  However, for $c_2 = (q_0, \omega, \gamma) \SLASHF (q_1, \omega,
  \gamma)$ the spine generators of $\gamma$ and of the output 
  of $q_1$ are different ($\overline{c}$ and $s$).  Hence it is a
  non-primary category and gets relabeled to~$\gamma$. 

  Concerning the lexicon, $c_1$ is a lexical category due to the fact
  that $(\bot, \varepsilon, q'_3) \in \TOPCATL$ can appear at the top
  of a spine as an initial category with $q'_3 \in F$ in its third
  component, while the appended $(q_0, \omega, \gamma)$ represents an
  initial configuration of~$\AUTOMATON$.  Similarly, $c_2$ is a
  well-formed secondary category of a rule and the third component of
  its target is in~$L_1$.  Therefore, it is an element of~$\TOPCATS$,
  which is a subset of the lexicon.

  Let us illustrate how the attachment of the side spine to the main
  spine is realized.  The lexicon contains $(q'_1, \upsilon, p'_1)
  \SLASHB (q_2, \upsilon, \alpha) \SLASHB (p_0, \chi, \beta)$, of
  which the first two atoms are responsible for performing a
  transition on the main spine.  This part cannot be modified since
  the rule system disallows it.  The target stores the final
  state~$p'_1$ of the side spine run in its third component.  The
  appended argument models the initial configuration of the side spine
  run starting in state $p_0$ with $\chi$ on the stack.
\end{example}

\subsection{Relating CCG Spines and Automaton Runs}
\label{sec:6}

We assume the symbols that were introduced above.
In particular, let $\mathcal A = (Q, \Delta, \Gamma, \delta, \tau, I,
F)$ be the pop-normalized MPDA, let $L_1 = \bigl\{ w \in
\Next(\SPINES(\mathcal G)) \mid \abs w = 1 \bigr\}$ be the short strings
not captured by $\mathcal A$,
and let $\mathcal G_{\AUTOMATON, L_1} = (\Delta_0, A, R, I', {\LEXICON})$ be the
constructed CCG.
We start with discussing the spines before we
move on to the discussion of how those spines attach to each other in
the next subsection.

\begin{lemma}
Every primary input spine of a derivation tree of~$\mathcal D(\mathcal
G_{\AUTOMATON, L_1})$ read from bottom to top is relabeled to a
string~$w \in \Next(\SPINES(\mathcal G))$.\label{lemma:spine_to_run} 

\label{lem:spines_in_next}
\end{lemma}
\begin{proof}
  We start with spines of length~$1$.  Their single category is
  obviously taken from the lexicon and thus either an initial
  atomic category~$a$ or a secondary input category~$ax$.  Both of
  those categories are relabeled to~$a_3 \in L_1 \subseteq
  \Next(\SPINES(\mathcal G))$.

  Now consider a primary input spine~$c_0 \cdots c_n$ with $n \geq
  1$.  We have to show that there is an accepting run of~$\AUTOMATON$
  corresponding to this spine.  We have already described the general
  form of these spines.  There exists a category~$bx$ and for each $i
  \in \{0, \dotsc, n\}$ there exist $m \in \nat$ and slashes $\seq
  \SLASH1m \in D$ as well as atoms~$\seq a1m \in A$ such that
  category~$c_i$ has the form~$bx \,\SLASH_1\, a_1 \cdots \,\SLASH_m\,
  a_m$.  In particular, we have $c_0 = bx \,\SLASH\, a$ for some
  $\SLASH \in D$ and $a \in A$ as well as $c_n = bx$.
  For better readability, we address the $j$-th component of atom
  $a_i$ by $a_{i,j}$ as an abbreviation for $(a_i)_j$, where $i \in [m]$
  and $j \in [3]$.
  We translate each category~$bx \,\SLASH_1\, a_1 \cdots \,\SLASH_m\,
  a_m$ to a configuration of~$\AUTOMATON$
  via the mapping $\text{conf} \colon \CATEGORIES{A} \to \textrm{Conf}_\AUTOMATON$
  in the following manner.
  \begin{align*}
    \text{conf}(bx \,\SLASH_1\, a_1 \cdots \,\SLASH_m\, a_m) = 
    \begin{cases}
      \langle b_3,\, \varepsilon \rangle & \text{if } m = 0 \\
      \langle a_{m,1},\, a_{m,2} \cdots a_{1,2} \rangle & \text{otherwise}
    \end{cases}
  \end{align*}
  In other words, the state of the configuration corresponding
  to~$c_n$ is the third component of the target~$b$, whereas all other
  categories~$\seq c0{n-1}$ store the state in the first component of
  the last argument.  The stack content is represented by the second
  components of the suffix~$\SLASH_1\, a_1 \cdots \,\SLASH_m\, a_m$.
  Each category relabels to the input symbol produced by its respective stored
  state.  Thus, if $\bigl(\text{conf}(c_0), \dotsc, \text{conf}(c_n)
  \bigr)$ is an accepting run of~$\AUTOMATON$, then it generates
  the same string that the spine is relabeled to.

  It remains to show that $\bigl(\text{conf}(c_0), \dotsc,
  \text{conf}(c_n) \bigr)$ is actually an accepting run.
  Since~$c_0$~is assigned by the lexicon and has a suffix behind~$bx$
  consisting of a single argument, whereas $c_n$~has an empty suffix, it is
  easy to see that $c_0$~and~$c_n$ correspond to an initial and a
  final configuration, respectively.  Hence we only need to prove that
  configurations corresponding to subsequent categories
  $c_i$~and~$c_{i+1}$ with $i \in \{0, \dotsc, n-1\}$ are connected by valid moves.
  To this end, we distinguish three cases based on the rule that is
  used to derive output category~$c_{i+1}$ from primary 
  category~$c_i$:
\begin{enumerate}
\item Let $r \in R_1^{\SLASH}$ with $\SLASH \in D$.  Moreover, let
  $c_i = bx \,\SLASH_1\, a_1 \cdots \,\SLASH_m\, a_m$ with $\SLASH_m =
  \SLASH$ and $c_{i+1} = bx \,\SLASH_1\, a_1 \cdots 
  \,\SLASH_{m-1} a_{m-1} \,\SLASH'_m\, a'_m$.   These categories
  correspond to configurations $\langle a_{m,1},\, a_{m,2} \cdots
  a_{1,2} \rangle$~and~$\langle a_{m,1}',\, a_{m,2}' a_{m-1,2}\cdots
  a_{1,2} \rangle$,
  respectively.  The definition of~$R_1^{\SLASH}$ implies
  $a_{m,2} = a_{m,2}'$ as well as the existence of the
  transition~$(a_{m,1}, \varepsilon, \varepsilon, a_{m,1}') \in
  \delta$ of~$\AUTOMATON$ that enables a valid move.

\item Let $r \in R_2^{\SLASH}$ with $\SLASH \in D$.  Moreover, let
  $c_i = bx \,\SLASH_1\, a_1 \cdots \,\SLASH_m\, a_m$ with $\SLASH_m =
  \SLASH$ and $c_{i+1} = bx \,\SLASH_1\, a_1 \cdots \,\SLASH_{m-1}
  a_{m-1} \,\SLASH_m'\, a_m' \,\SLASH_{m+1}'\, a_{m+1}'$.
  The corresponding configurations are $\langle a_{m,1},\, a_{m,2} \cdots
  a_{1,2}\rangle$ and $\langle a_{m+1,1}',\, a_{m+1,2}', a_{m,2}' a_{m-1,2}
  \cdots a_{1,2} \rangle$,
  respectively.  The definition of~$R_2^{\SLASH}$ implies $a_{m,2} =
  a_{m,2}'$ as well as the existence of the transition $(a_{m,1},
  \varepsilon, a_{m+1,2}', a_{m+1,1}') \in \delta$, which again
  enables a valid move.

\item Let $r \in R_3^{\SLASH}$ with $\SLASH \in D$.  Moreover, let
  $c_i = bx \,\SLASH_1\, a_1 \cdots \,\SLASH_m\, a_m$ with $\SLASH_m =
  \SLASH$ and $c_{i+1} = bx \,\SLASH_1\, a_1 \cdots \,\SLASH_{m-1}\,
  a_{m-1}$.  We distinguish two subcases.  First, let $m > 1$.
  These categories correspond to configurations $\langle a_{m,1},\, a_{m,2}
  \cdots a_{1,2} \rangle$ and $\langle a_{m-1,1},\, a_{m-1,2} \cdots
  a_{1,2} \rangle$,
  respectively.  Due to the coupling of neighboring arguments, we have
  $\RETURN(a_{m,2}) = a_{m-1,1}$, which ensures that we reach the
  correct target state.  Since $\AUTOMATON$~is pop-normalized, the
  target state of the pop-transition is completely determined by the
  popped symbol~$a_{m,2}$.  The definition of~$R_3^{\SLASH}$ implies
  that the transition $(a_{m,1}, a_{m,2}, \varepsilon,
  a_{m-1,1}) \in \delta$ exists, which again makes the move valid.

  Now let $m = 1$, which yields $c_i = bx \,\SLASH\, a_1$ and $c_{i+1}
  = bx$.  The corresponding configurations are $\langle a_{1,1},\, a_{1,2}
  \rangle$ and $\langle b_3,\, \varepsilon \rangle$,
  where $\RETURN(a_{1,2}) = b_3$ because the
  initial stack symbol~$a_{1,2}$ was assigned by the lexicon and
  cannot be modified without removing the argument.  The definition
  of~$R_3^{\SLASH}$ implies the existence of the transition $(a_{1,1},
  a_{1,2}, \varepsilon, b_3) \in \delta$, which also concludes this
  subcase.
\end{enumerate}

We have seen that each spine~$c_0 \cdots c_n$ with~$n \geq 1$
corresponds to an accepting run~$\text{conf}(c_0) \vdash_{\AUTOMATON}
\cdots \vdash_{\AUTOMATON} \text{conf}(c_n)$ whose output~$w \in
L(\AUTOMATON) \subseteq \Next(\SPINES(\mathcal G))$ coincides with the
relabeling of the spine.
\end{proof}

We now turn our attention to the inverse direction.
More precisely, we will show that, given a string~$w = w_0
\cdots w_n \in \Next(\SPINES(\mathcal G))$, we can find a primary spine~$c_0
\cdots c_n$ of~$\mathcal G_{\AUTOMATON, L_1}$
(i.e., a sequence of primary categories starting at a category
that belongs to~$\LEXICON(\Delta_0)$ and ending in a non"-primary category)
that gets relabeled to it.
Further, for this
spine we have some freedom in the selection of the topmost category,
so that for every valid secondary category or initial
category, in which the third component of the target outputs~$w_n$,
we can choose either this category or another category that differs only in the third component, but still outputs~$w_n$.
Additionally, 
the third component of the last argument in the suffix (so in all
categories except for $c_n$)
can be chosen freely from the set of strings of length~$1$ or
final states with the correct spine generator.  This will be of great
importance when we combine these spines to a complete derivation
tree.

\begin{lemma}
  For each string~$w \in L_1$ and category $c \in \{ ax \in
  {\LEXICON}(\Delta_0) \mid a_3 = w \}$ there is a primary spine
  of~$\mathcal G_{\AUTOMATON, L_1}$ that is labeled by~$c$ and
  relabeled to~$w$.
\label{lem:l1_in_spines}
\end{lemma}
\begin{proof}
  The set~$C  = \{ ax \in {\LEXICON}(\Delta_0) \mid a_3 = w \}$ is
  clearly a subset of~$\LEXICON(\Delta_0)$
  and each~$ax \in C$ is either a secondary category or an initial
  category by the construction of~$\LEXICON$.
  In either case $ax$ is relabeled to~$a_3 = w$ and cannot be modified.
  Consequently, these
  categories themselves constitute complete primary spines of
  length~$1$.
\end{proof}

\begin{lemma}
\label{lem:la_in_spines}
  Assume we are given an accepting run~$\bigl(\langle q_0, \gamma_0
  \rangle, \dotsc, \langle q_n,
  \varepsilon \rangle  \bigr)$ of $\mathcal A$, well"-formed
  category~$\hat{c} \in \{ ax \mid ax \,\SLASH\, b \in
  {\LEXICON}(\Delta_0), a_3 = q_n \}$, and $\seq e0{n-1} \in F \, \cup\, L_1$
  such that~$\GENERATOR(e_i) = \COMB(q_i)$ for all~$i \in \{0, \dotsc,
  n-1 \}$.
  Then there exists a primary spine~$c_0 \dotsm c_{n-1} \hat c$
  with well"-formed categories~$\seq c0{n-1}$ such that it relabels to $\rho(c_0) \dotsm
  \rho(c_{n-1}) \rho(\hat c) = \tau(q_0) \dotsm \tau(q_n)$ 
  and for all~$i \in \{0, \dotsc, n-1\}$ the category~$c_i$ starts
  with category~$\hat{c}$, ends with an argument~$\SLASH \, a$
  with~$a_3 = e_i$, and $\abs{\gamma_i} = \ARITY(c_i) - \ARITY(\hat
  c)$.
\end{lemma}
\begin{proof}
We will describe how to construct the primary spine~$\word
c0{n-1}\hat c$ by induction on~$i \in \{0, \dotsc, n\}$ that
additionally obeys the following invariants:
\begin{enumerate}[label=(\roman*)]
\item
The categories are well"-formed.
\item
Subsequent arguments 
$\SLASH_{j-1}\, a_{j-1} \,\SLASH_j\, a_j$ in the suffix
(i.e.\ the argument sequence after $\hat{c}$ in each category
$\hat{c} \,\SLASH_1\, a_1 \dotsm \,\SLASH_m\, a_m$)
are coupled in such a way that~$\RETURN((a_j)_2) = (a_{j-1})_1$ for
all~$j \in [m] \setminus \{1\}$.
\end{enumerate}

We already noted that the suffix of a category stores the stack in the
second components such that the second component of the last argument
contains the topmost stack symbol. Also note that all states $q_0, \dots,
q_n$ have the same spine generator.

In the induction base, we let $c_0 = \hat c \,\SLASH\, (q_0, \gamma_0,
e_0)$ with $\mathord{\SLASH} = \SLASHDIR(q_0)$.  We note
that~$\RETURN(\gamma_0) = \TARGET(\hat c)_3 = q_n$.  Obviously $c_0
\in \LEXICON(\Delta_0)$ and $\rho(c_0) = \STATEOUTPUT(q_0)$.  Clearly,
$c_0$ obeys the invariants, and has the right arity.

In the induction step, we assume that $c_{i-1}$ already fulfills the
conditions, and we let $(q_{i-1}, \gamma, \gamma',  q_i) \in \delta$ be a
transition that permits the move $\langle q_{i-1}, \gamma_{i-1}\rangle
\vdash_{\mathcal A} \langle q_i, \gamma_i \rangle$.
We again distinguish three cases
for the construction of a suitable category~$c_i$:

\begin{enumerate}
\item \emph{Ignore stack:} Suppose that $\gamma = \gamma' =
  \varepsilon$.  Let $c_{i-1} = \hat{c} \,\SLASH_1\, a_1 \dotsm
  \,\SLASH_m\, a_m$.  We apply a rule of~$R_1^{\SLASH_m}$ and obtain
  $c_i = \hat{c} \,\SLASH_1\, a_1 \dotsm \,\SLASH_{m-1}\, a_{m-1}
  \,\SLASH\, b$, where $\SLASH = \SLASHDIR(q_i)$ and $b = \bigl(q_i,
  (a_m)_2, e_i \bigr)$.
  Since $b_1 = q_i$, the category~$c_i$ gets relabeled
  to~$\STATEOUTPUT(q_i)$.  The stack symbol is not changed, so the
  category is well-formed and subsequent arguments are still coupled.
  Additionally, neither the stack size nor the arity of the category have changed.

\item \emph{Push symbol:} Suppose that $\gamma' \neq \varepsilon$.    
  Moreover, let $j \in \{i+1, \dotsc, n\}$ be minimal such that
  $\gamma_j = \gamma_{i-1}$ (i.e., $j$ is the index of the configuration in
  which the stack symbol~$\gamma'$ was removed again).  Finally, let
  $c_{i-1} = \hat{c} \,\SLASH_1\, a_1 \dotsm \,\SLASH_m\, a_m$.
  Clearly, we apply a rule of~$R_2^{\SLASH_m}$ to obtain
  $c_i = \hat{c} \,\SLASH_1\, a_1 \dotsm \,\SLASH_{m-1}\, a_{m-1}
  \,\SLASH\, b \, \SLASH'\, b'$ with $\mathord{\SLASH} =
  \SLASHDIR(q_j)$, $\mathord{\SLASH'} = \SLASHDIR(q_i)$, $b' = (q_i,
  \gamma', e_i)$, and $b = \bigl(q_j, (a_m)_2, e_j\bigr)$.  Note
  that~$q_j = \RETURN(\gamma')$.  Hence the category~$c_i$ gets
  relabeled to~$\STATEOUTPUT(q_i)$.  The mentioned conditions ensure
  that $c_i$~is well-formed and obeys the second invariant.
  The increase in
  stack size is properly accounted for by an increased arity of~$c_i$.

\item \emph{Pop symbol:}  Suppose that $\gamma \neq \varepsilon$.  We
  further distinguish between the cases $i < n$~and~$i = n$.  We start with
  $i < n$.  Suppose that $c_{i-1} = \hat{c} \,\SLASH_1\, a_1 \dots
  \,\SLASH_m\, a_m$.  Note that $m \geq 2$.  Since $c_{i-1}$ obeys the
  invariants, subsequent arguments in the suffix are coupled, so
  we have~$(a_{m-1})_1 = \RETURN((a_m)_2) = \RETURN(\gamma)
  = q_i$.  Additionally, $(a_{m-1})_3 = e_i$ as prepared in the
  corresponding push transition.  We apply a rule of~$R_3^{\SLASH_m}$
  to obtain $c_i = \hat{c} \,\SLASH_1\, a_1 \dots
  \,\SLASH_{m-1}\, a_{m-1}$, which is trivially well-formed and still
  obeys the second invariant.
  It relabels to~$\tau(q_i)$ as required due
  to~$(a_{m-1})_1 = q_i$.  The stack size and arity both decrease
  by~$1$.

  For~$i = n$ we have $c_{n-1} = \hat c \,\SLASH_1\, a_1$ since the stack
  size is necessarily~$1$.  We also apply a rule of~$R_3^{\SLASH_1}$
  and obtain the category~$\hat c$.  This category is trivially
  well-formed and also trivially fulfills the second invariant.
  Additionally, it relabels
  to~$\tau(q_n)$ since it is a secondary category
  and~$\TARGET(\hat c)_3 = q_n$.
\end{enumerate}
\end{proof}

We observe that the restrictions on arguments in categories~$c_i$
for~$i \in \{0, \dots, n\}$ also hold for the arguments (and for the
target) of the secondary categories that are needed to perform
the category transformations corresponding to the automaton run.
All of these secondary categories themselves can be chosen as the
category at the top of an appropriate primary spine (unless the third
component of the target constitutes an unreachable state of~$\AUTOMATON$).
This will be relevant in the next step, in which we combine the spines to obtain a complete derivation tree.

\subsection{Combining Spines}
We continue to use the introduced symbols.  Moreover, let
$
\mathcal{D'}(\mathcal G_{\AUTOMATON, L_1}) =
\{ t \in \mathcal{D}(\mathcal G_{\AUTOMATON, L_1}) \mid t(\varepsilon) \in 
\TOPCATS \cup \TOPCATL \}
$.
In other words, these are exactly the derivation trees whose root
nodes are labeled by top"-of"-spine categories.
We will show that $\rho \bigl(\mathcal{D'}(\mathcal G_{\AUTOMATON,
  L_1}) \bigr) = {\FORESTS} \bigl(\Next(\SPINES(\mathcal G)) \bigr)$. 

\begin{lemma}
  \label{lem:dg_in_next}
  $\rho \bigl(\mathcal{D'}(\mathcal G_{\AUTOMATON, L_1}) \bigr)
  \subseteq {\FORESTS} \bigl(\Next(\SPINES(\mathcal G)) \bigr)$ 
\end{lemma}

\begin{proof}
  We prove the statement by induction on the size of~$t \in
  \mathcal{D'}(\mathcal G_{\AUTOMATON, L_1})$.  Let $\word c0n$ be the
  primary spine of~$t$ that starts at a lexicon entry~$c_0 \in
  \LEXICON(\Delta_0)$ and ends at the root (i.e.~$c_n =
  t(\varepsilon)$).  By Lemma~\ref{lem:spines_in_next}, this spine
  gets relabeled to a string~$w = \word w0n \in \Next
  \bigl(\SPINES(\mathcal G)\bigr)$.  Except for the root
  category~$c_n$, each of the spinal categories~$\seq c0{n-1}$ gets
  combined with a secondary category that is itself the root of
  a subtree~$t' \in \mathcal{D'}(\mathcal G_{\AUTOMATON, L_1})$.
  Since $t'$~is a proper subtree of~$t$, we can utilize the induction
  hypothesis to conclude that $\rho(t') \in {\FORESTS}
  \bigl(\Next(\SPINES(\mathcal G))\bigr)$.  It remains to show that
  each such tree fulfills the requirements necessary to attach it to
  the spine.  Suppose that the primary category is~$c_i = ax
  \,\SLASH\, b$, so it can only be combined with a secondary 
  category of the form~$by$, which gets relabeled
  to~$\rho(by) = \STATEOUTPUT(b_3)$.  Suppose further that $\rho(c_i) =
  \bigl(\langle \sigma', n'_1, n'_2 \rangle, \langle \sigma, n_1, n_2 \rangle
  \bigr)$.  Clearly, $\GENERATOR(b_3) = \COMB(b_1)$, where
  $\GENERATOR(b_3)$~is the spine generator at the root
  of~$\rho(t')$, and $\COMB(b_1) = n'_{3-d(\sigma')}$ is the generator
  of the non-spinal child of the succeeding parent
  symbol~$\rho(c_{i+1})$.  Since they coincide, the attachment of~$\rho(t')$
  is possible and the directionality of the attachment is~$3 -
  d(\sigma')$, which is guaranteed by the requirement $\SLASH =
  \SLASHDIR(b)$ for argument $\SLASH\, b$.  We conclude that all
  attachments of subtrees are consistent with the definition
  of~${\FORESTS} \bigl(\Next(\SPINES(\mathcal G)) \bigr)$.  Thus,
  $\rho(t) \in {\FORESTS} \bigl(\Next(\SPINES(\mathcal G)) \bigr)$.
\end{proof}

\begin{lemma}
  \label{lem:fnext_in_dg}
  ${\FORESTS} \bigl(\Next(\SPINES(\mathcal G)) \bigr) \subseteq \rho
  \bigl(\mathcal{D'}(\mathcal G_{\AUTOMATON, L_1}) \bigr)$
\end{lemma}
\begin{proof}
  Indeed we prove the following statement for all~$t \in {\FORESTS}
  \bigl(\Next(\SPINES(\mathcal G))\bigr)$.  If $\abs t = 1$, then for
  each~$c \in \{ax \in \TOPCATS \mid a_3 = t\}$ there is a tree~$t'
  \in \mathcal{D'}(\mathcal G_{\AUTOMATON, L_1})$ such that
  $t'(\varepsilon) = c$~and~$\rho(t') = t$.  If $\abs t > 1$, then for
  each~$c \in \{ax \in \TOPCATL \mid a_3 = q_n\}$, where $q_n$ is the
  final state of an accepting run of~$\AUTOMATON$ corresponding to the
  main spine of~$t$, there is a tree~$t' \in \mathcal{D'}(\mathcal
  G_{\AUTOMATON, L_1})$ such that $t'(\varepsilon) = c$~and~$\rho(t') = t$.
  We perform an induction on the size of~$t$.

  In the induction base, $t$~consists of a single node~$t \in L_1$.  By
  Lemma~\ref{lem:l1_in_spines}, all categories~$c \in \{ ax \in
  {\LEXICON}(\Delta_0) \mid a_3 = t \} = \{ ax \in \TOPCATS \mid a_3 =
  t \}$ are complete primary spines that get relabeled to~$t$.  By the
  definition of the lexion, this set is nonempty for all~$t \in L_1$.
  Thus, $t \in \rho(\mathcal{D'}(\mathcal G_{\AUTOMATON, L_1}))$.

  In the induction step, let $t \in {\FORESTS}
  \bigl(\Next(\SPINES(\mathcal G))\bigr)$ be a tree with~$\abs t > 1$.
  We identify the main spine labeled by $w = w_0 \dotsm w_n \in \Next
  \bigl(\SPINES(\mathcal G) \bigr)$ that was used to create~$t$.  This
  string~$w$ is generated by an accepting run~$(\langle q_0,
  \alpha_0\rangle, \dots, \langle q_n, \varepsilon\rangle)$
  of $\AUTOMATON$.  Likewise, there exists a primary 
  spine~$\word c0n$ of~$\mathcal{G}_{\AUTOMATON, L_1}$ that gets
  relabeled to~$w$ and we can choose the category~$c_n$ at the top of
  the spine freely from the set
  \[ \{ ax \mid ax \,\SLASH\, b \in {\LEXICON}(\Delta_0), a_3 = q_n \}
    = \{ ax \in \TOPCATL \mid a_3 = q_n\} \]
  according to Lemma~\ref{lem:la_in_spines}.  Similarly, we can also
  freely choose the third component of the last argument of each
  category~$\seq c0{n-1}$.  Consider an arbitrary~$0 \leq i \leq n-1$,
  and let $w_i = (\langle \sigma', n_1', n_2' \rangle, \langle \sigma,
  n_1, n_2 \rangle)$; the case of~$w_i = (\langle \sigma', n_1', n_2'
  \rangle, \langle \sigma, n \rangle)$ is analogous.  Below the
  parent node that is labeled by~$w_{i+1}$,
  in the direction~$3-d(\sigma')$, a tree~$t' \in
  {\FORESTS} \bigl(\Next(\SPINES(\mathcal G))\bigr)$ with
  spine generator~$n'_{3-d(\sigma')}$ stored in its root
  label is attached.  Let~$q$ be the final state of an accepting run
  of~$\mathcal A$ for the main spine of~$t'$ when $\abs{t'} > 1$ or
  let $q = t'$ when $\abs{t'} = 1$.  Moreover, suppose that~$c_i = ax
  \,\SLASH\, b$ with~$b_1 = q_i$, so the required secondary 
  category has the shape~$by$.  By the induction hypothesis, there
  exists $t'' \in \mathcal{D'}(\mathcal G_{\AUTOMATON, L_1})$ such that
  $t''(\varepsilon) = by \in \TOPCATS \cup \TOPCATL$, $\rho(t'') = t'$, and $b_3 = q$.
  This choice of~$b_3$ is possible for~$c_i$ provided
  that~$\GENERATOR(b_3) = \COMB(b_1)$, which is the case since~$\COMB(b_1) = \COMB(q_i) =
  n'_{3-d(\sigma')} =  \GENERATOR(q) = \GENERATOR(b_3) $.
  The directionality~$3-d(\sigma')$ of the attachment of~$t''$
  is guaranteed by the relationship $\SLASH = \SLASHDIR(b)$, which
  holds since $c_i$ is well-formed.  In conclusion, for each attached
  subtree~$t' \in {\FORESTS} \bigl(\Next(\SPINES(\mathcal G))\bigr)$
  of~$t$ we can find a suitable~$t'' \in \mathcal{D'}(\mathcal
  G_{\AUTOMATON, L_1})$ whose root category can be combined with the
  neighboring primary category of the spine~$\word c0n$.
  Putting the primary spine and the derivation trees for
  subtrees together again yields a tree in~$\mathcal{D'}(\mathcal
  G_{\AUTOMATON, L_1})$.  Its root~$c_n$ can be chosen freely from the
  desired set~$\{ ax \mid ax \,\SLASH\, b \in {\LEXICON}(\Delta_0),\,
  a_3 = q_n\}$.
\end{proof}

\begin{lemma}
  ${\FORESTS} \bigl(\Next(\SPINES(\mathcal G)) \bigr)_S =
  \TREELANG_\rho(\mathcal G_{\AUTOMATON, L_1})$
\end{lemma}
\begin{proof}
  Recall that~$\TREELANG_{\rho}(\mathcal G_{\AUTOMATON, L_1}) = \{
  \rho(t) \in \TREES{\Delta,\emptyset}(\Delta) \mid t \in
  \DEVTREES(\mathcal G_{\AUTOMATON, L_1}),\, t(\varepsilon) \in I'\}$,
  which yields ${\TREELANG}_\rho(\mathcal G_{\AUTOMATON, 
    L_1}) \subseteq \rho \bigl(\mathcal{D'}(\mathcal G_{\AUTOMATON,
    L_1}) \bigr)$.  Also recall that the initial atomic
  categories of~$\mathcal G_{\AUTOMATON,  L_1}$ are~$I' = \{ (\bot,
  \varepsilon, f) \in A \mid \GENERATOR(f) = S\}$.
  Let $t \in
  {\FORESTS} \bigl(\Next(\SPINES(\mathcal G)) \bigr)_S$.  By
  Lemma~\ref{lem:fnext_in_dg}, there is a tree~$t' \in
  \DEVTREES(\mathcal G_{\AUTOMATON,  L_1})$ with~$\rho(t') = t$, whose
  root category can be any category from $\{ ax \in \TOPCATL \mid a_3
  = f\}$, where $f$~is the final state of an accepting run for the
  main spine of~$t$, or from $\{ax \in \TOPCATS \mid a_3 = t\}$, if
  $t$~consists of a single node.  Hence we can select the
  category~$t'(\varepsilon) = (\bot, \varepsilon, f)$ in
  the former and~$t'(\varepsilon) = (\bot, \varepsilon,
  t)$ in the latter case.
  Since both of these categories are
  initial, we obtain $\rho(t') \in {\TREELANG}_\rho(\mathcal G_{\AUTOMATON,
    L_1})$.

Now let~$t \in {\TREELANG}_\rho(\mathcal G_{\AUTOMATON, L_1})$.
Hence there is a tree $t'\in \DEVTREES(\mathcal G_{\AUTOMATON, L_1})$
with $\rho(t') = t$ and $t'(\varepsilon) \in
I' = \{ (\bot, \gamma, f) \in A \mid \GENERATOR(f) = S\}$.
By Lemma~\ref{lem:dg_in_next}, we also have $t
\in {\FORESTS} \bigl(\Next(\SPINES(\mathcal G)) \bigr)$.  Because
$\rho \bigl(t'(\varepsilon) \bigr) = t(\varepsilon) = f$ with~$\GENERATOR(f) = S$,
we obtain $t \in {\FORESTS} \bigl(\Next(\SPINES(\mathcal G)) \bigr)_S$. 
\end{proof}

Together with Corollary~\ref{cor:spines},
this concludes the proof of the following main theorem.

\begin{theorem}
  \label{theorem:spine_grammar_to_ccg}
  Given a spine grammar~$\mathcal G$, we can construct a CCG~$\mathcal
  G'$ that can generate $\TREELANG(\mathcal G)$.
\end{theorem}

\section{Strong Equivalence}

In this section we will show that CCG and TAG are strongly equivalent
modulo relabeling. We will also cover the implications regarding the role
of $\varepsilon$-entries, rule degree, and the use of first"-order categories.

For the converse inclusion of Theorem~\ref{theorem:spine_grammar_to_ccg}
we utilize a result by Kuhlmann, Maletti, and
Schiffer~\cite[Theorem~29]{kuhmalsch2021}.
It states that for every
CCG~$\mathcal G$ there exists an sCFTG that generates the rule trees
of~$\mathcal G$.  While derivation trees are labeled by categories,
\emph{rule trees} are labeled by lexicon entries at leaves and by
applied rules (instead of the output category) at inner nodes.   
Rule trees are a natural encoding of derivation trees using only a
finite set of labels.
As each rule indicates the target and last
argument of its output category, rule trees can be relabeled in the
same manner as derivation trees.
For completeness' sake we restate
the following definition~\cite[Definition 22]{kuhmalsch2021}.

\begin{definition}
  Let $\mathcal G = (\Sigma, A, R, I, {\LEXICON})$ be a \CCG\@ and
  $\mathrm T = \TREES{R,\emptyset}(\LEXICON(\Sigma))$.  A tree $t \in 
  \mathrm T$ is a \emph{rule tree} if $\type(t) \in I$, where
  the partial map $\mathord{\type} \colon \mathrm T \to \CATEGORIES A$
  is inductively defined by 
  \begin{enumerate}[label=(\roman*)]
  \item
  $\type(a) = a$ for all lexicon entries
  $a \in {\LEXICON}(\Sigma)$,
  \item
  $\type \bigg(\cfrac{axy}{ax
    \SLASHF b \quad by}(t_1, t_2) \bigg) = a z y$ for all
  trees $t_1, t_2 \in \mathrm T$ with $\type(t_1) = a z
  \SLASHF b$ and $\type(t_2) = by$, and
  \item
  $\type \bigg(
  \cfrac{axy}{by \quad ax \SLASHB b}(t_1, t_2) \bigg) =
  a z y$ for all trees $t_1, t_2 \in \mathrm T$ with
  $\type(t_1) = by$ and $\type(t_2) = a z \SLASHB b$.
  \end{enumerate}
  The set of all rule trees of~$\mathcal G$ is denoted by~$\mathcal
  R(\mathcal G)$.
\end{definition}

This leads us to the second main theorem.

\begin{theorem}
  CCGs and sCFTGs are strongly equivalent up to relabeling.
\end{theorem}

\begin{proof}
  Let~$\mathcal G$ be a CCG. Then its rule tree
  language~$\RULETREES(\mathcal G)$ can be generated by an
  sCFTG~$\mathcal G'$~\cite[Theorem~29]{kuhmalsch2021}.
  The tree language~$\TREELANG_\rho(\mathcal G)$ accepted by~$\mathcal G$
  is the set of derivation trees~$\DEVTREES(\mathcal G)$ relabeled
  by~$\rho$.  The relabeling~$\rho$ can be transferred to the rule tree
  language~$\RULETREES(\mathcal G)$ since it only depends on the
  target and the last argument of each category, which can both be
  figured out by looking at the output category of the rule label of
  the respective rule tree node.  Conversely, given an sCFTG~$\mathcal
  G$ we can first convert it into an equivalent spine grammar and then
  construct a CCG that is equivalent (up to relabeling) to~$\mathcal
  G$ by Theorem~\ref{theorem:spine_grammar_to_ccg}.
\end{proof}

Kepser and Rogers~\cite{kepser2011equivalence} proved that TAGs and sCFTGs are strongly
equivalent, which shows that they are also strongly equivalent (up to
relabeling) to CCGs. 

\begin{corollary}
  CCGs and TAGs are strongly equivalent up to relabeling.
\end{corollary}

Clearly, from strong equivalence we can conclude weak equivalence as
well (without the relabeling since the lexicon provides the
relabeling).  Weak equivalence was famously proven by
Vijay"-Shanker and Weir~\cite{vijayshanker1994equivalence}, but Theorem~3~of
Kuhlmann, Koller, and Satta~\cite{kuhlmann2015lexicalization} shows that
the original construction is incorrect.  However,
Weir~\cite{weir1988characterizing} provides an alternative construction
and proof.  Our contribution provides a stronger form (and proof) of
this old equivalence result.  It avoids the $\varepsilon$"~entries
that the original construction heavily relies on.  An
$\varepsilon$"~entry is a category assigned to the empty string; these
interspersed categories form the main building block in the original
constructions.  The necessity of these
$\varepsilon$"~entries~\cite{vijayshanker1994equivalence} is an
interesting and important question that naturally arises and has been
asked by Kuhlmann, Koller, and Satta~\cite{kuhlmann2015lexicalization}.
We settle this question and demonstrate that they can be avoided. 

\begin{corollary}
  CCGs and TAGs are weakly equivalent. Moreover, CCGs with
  $\varepsilon$"~entries and CCGs generate the same
  ($\varepsilon$"~free) languages.
\end{corollary}

\begin{proof}
  The weak equivalence of CCG and TAG is clear from the previous
  corollary.  Similarly, each $\varepsilon$"~free language generated
  by a CCG can trivially also be generated by a CCG with
  $\varepsilon$"-entries.  For the converse direction, let $\mathcal
  G$ be a CCG with $\varepsilon$"-entries.  We convert it into a spine
  grammar~$\mathcal G'$ in normal form accepting the rule tree
  language of~$\mathcal G$.  A standard $\varepsilon$"-removal construction
  yields a weakly equivalent spine grammar~$\mathcal G''$, which can
  be converted into a strongly equivalent CCG (up to relabeling).
  This constructed CCG accepts the same $\varepsilon$-free string
  language as the original CCG~$\mathcal G$ that utilized
  $\varepsilon$"-entries.
\end{proof}

The tree expressive power of CCGs with restricted rule degrees has
already been investigated by Kuhlmann, Maletti, and
Schiffer~\cite{kuhlmann2019treegenerative,kuhmalsch2021}.  It
has been shown that 0"~CCGs accept a proper subset of the regular tree
languages~\cite{gecste97}, whereas 1"~CCGs accept exactly the regular
tree languages.  It remained open whether there is a~$k$ such that
$k$"~CCGs and $(k+1)$"~CCGs have the same expressive power.  Our
construction establishes that $2$"~CCGs are as expressive as $k$"~CCGs
for arbitrary $k \geq 2$.  Another consequence of our construction is
that first-order categories are sufficient.

\begin{corollary}
  \label{cor:2ccg}
  \label{cor:firstorder}
  $2$"~\CCG{s} with first-order categories have the same expressive
  power as $k$"~\CCG{s} with $k > 2$.
\end{corollary}

\begin{proof}
  We only argue the nontrivial inclusion.  Let $\mathcal G$~be a CCG
  whose categories have arbitrary order and whose rules have degree at
  most~$k$.  We construct the sCFTG~$\mathcal G'$ generating the rule
  tree language~$\RULETREES(\mathcal G)$.  Next, we construct the
  CCG~$\mathcal G''$ that generates the same tree language
  as~$\mathcal G'$ modulo relabeling.  By construction, $\mathcal
  G''$~uses only first"-order categories and rules with rule-degree at
  most~$2$.  As already argued, the rule trees can be relabeled to the
  relabeled tree language generated by~$\mathcal G$.
\end{proof}

\section{Conclusion}
We presented a translation from spine grammar to CCG.  Due to
  the strong equivalence of spine grammar and
  TAG~\cite{kepser2011equivalence}, we can also construct a strongly
  equivalent CCG for each TAG.  Together with the translation from CCG
  to sCFTG~\cite{kuhlmann2019treegenerative,kuhmalsch2021}, this proves the strong 
equivalence of TAG and CCG, which means that both formalisms generate
the same derivation trees modulo relabeling.
Our construction uses CCG rules of degree at most~$2$, only
first-order categories, lexicon entries of arity at most~$3$, and no
$\varepsilon$-entries in the lexicon.  Such CCGs thus have full
expressive power.  Avoiding $\varepsilon$-entries is particularly
interesting because they violate the Principle of
Adjacency~\cite[p.~54]{steedman2000syntactic}, which is a fundamental
linguistic principle underlying CCG and requires that all combining
categories correspond to phonologically realized counterparts in the
input and are string-adjacent.  Their elimination is performed by
trimming them from the sCFTG obtained from a CCG with
$\varepsilon$-entries and translating the trimmed sCFTG back to a CCG
using our construction.

Translating CCG to sCFTG~\cite{kuhlmann2019treegenerative,kuhmalsch2021}
yields sCFTGs whose size is exponential in a CCG-specific constant,
which depends on the maximal arity of secondary categories and of lexicon entries.
Note that the maximal arity of secondary categories can be higher than the
rule degree because it is also affected by the
maximal arity of categories in lexical arguments.
Our construction increases the grammar size only polynomially, which can be verified
for each step.  Overall, a $k$-CCG can be converted to an equivalent
$2$-CCG without $\varepsilon$-entries in time and space exponential
in the grammar"-specific constant and polynomial in the size of the grammar.
However, we expect that the construction can be improved to be
exponential only in the maximum rule degree~$k$, as this runtime can be achieved for
CCG parsing~\cite{schkuhsat22}, which is a task that is closely related.

\bibliographystyle{plain}
\bibliography{references}

\end{document}